\documentclass{article} 
\usepackage{iclr2027_conference,times}
\usepackage{amssymb}
\usepackage{amsthm}
\usepackage{bbm}
\usepackage{mathtools}
\usepackage{physics}
\usepackage{thmtools}
\usepackage{thm-restate}
\usepackage{algorithm}
\usepackage{algpseudocode}
\usepackage[utf8]{inputenc} 
\usepackage[T1]{fontenc}    
\usepackage{hyperref}       
\usepackage{url}            
\usepackage{booktabs}       
\usepackage{multirow}
\usepackage[version=4]{mhchem}
\usepackage{nicefrac}       
\usepackage{microtype}      
\usepackage{xcolor}         
\usepackage{graphicx}
\usepackage{subcaption}
\usepackage{wrapfig}         
\usepackage{tikz}
\usetikzlibrary{positioning, arrows.meta, fit, calc, backgrounds}
\usepackage[capitalize,noabbrev]{cleveref}
\usepackage{etoc}            
\usepackage{fontawesome5}    
\definecolor{markyes}{HTML}{2E9E73}
\definecolor{markno}{HTML}{E4572E}
\newcommand{\cmark}{\textcolor{markyes}{\faCheckCircle}}
\newcommand{\xmark}{\textcolor{markno}{\faTimesCircle}}

\definecolor{revcolor}{HTML}{8B0000}
\newcommand{\rev}[1]{{\ificlrfinal\else\color{revcolor}\fi#1}}
\newenvironment{revblock}{\ificlrfinal\else\color{revcolor}\fi}{}
\newcommand{\tnote}[1]{\rlap{$^{#1}$}}

\usepackage{amsmath,amsfonts,bm}

\def\eqref#1{equation~\ref{#1}}

\def\1{\bm{1}}

\def\rmC{{\mathbf{C}}}

\def\rmS{{\mathbf{S}}}

\def\va{{\bm{a}}}
\def\vb{{\bm{b}}}
\def\vc{{\bm{c}}}
\def\vd{{\bm{d}}}

\def\vj{{\bm{j}}}
\def\vk{{\bm{k}}}
\def\vl{{\bm{l}}}
\def\vm{{\bm{m}}}
\def\vn{{\bm{n}}}

\def\vq{{\bm{q}}}
\def\vr{{\bm{r}}}
\def\vs{{\bm{s}}}
\def\vt{{\bm{t}}}

\def\vv{{\bm{v}}}

\def\vx{{\bm{x}}}
\def\vy{{\bm{y}}}
\def\vz{{\bm{z}}}

\def\mA{{\bm{A}}}
\def\mB{{\bm{B}}}
\def\mC{{\bm{C}}}

\def\mF{{\bm{F}}}
\def\mG{{\bm{G}}}
\def\mH{{\bm{H}}}

\def\mJ{{\bm{J}}}
\def\mK{{\bm{K}}}
\def\mL{{\bm{L}}}

\def\mP{{\bm{P}}}
\def\mQ{{\bm{Q}}}
\def\mR{{\bm{R}}}
\def\mS{{\bm{S}}}
\def\mT{{\bm{T}}}
\def\mU{{\bm{U}}}

\def\mW{{\bm{W}}}

\def\mY{{\bm{Y}}}

\def\mSigma{{\bm{\Sigma}}}

\DeclareMathAlphabet{\mathsfit}{\encodingdefault}{\sfdefault}{m}{sl}
\SetMathAlphabet{\mathsfit}{bold}{\encodingdefault}{\sfdefault}{bx}{n}

\newcommand{\xc}{{\text{xc}}}

\newcommand{\vext}{{v _\text{ext}}}
\newcommand{\rr}{{\vr}}
\newcommand{\RR}{{\mR}}
\newcommand{\diag}{{\text{diag}}}

\theoremstyle{plain}
\newtheorem{theorem}{Theorem}[section]
\newtheorem{proposition}[theorem]{Proposition}

\theoremstyle{definition}
\newtheorem{definition}[theorem]{Definition}
\newtheorem{assumption}[theorem]{Assumption}
\theoremstyle{remark}

\iclrpreprint    
\newcommand{\coderepo}{\ificlrfinal\url{https://github.com/andresguzco/gs-dft}\else\url{https://anonymous.4open.science/r/gs-dft-26FC}\fi}
\newcommand{\coderepokind}{\ificlrfinal a public repository\else an anonymized repository\fi}
\newcommand{\papertitle}{Scaling Density Functional Theory with Gaussian Splatting}
\title{Scaling Density Functional Theory with \\Gaussian Splatting}

\author{%
  \textbf{Andr\'es Guzm\'an-Cordero\textsuperscript{1,\,2} \quad
  Cindy Zhang\textsuperscript{3} \quad
  Majdi Hassan\textsuperscript{1,\,2}} \\[0.2em]
  \textbf{Marta Skreta\textsuperscript{1,\,2} \quad
  Kirill Neklyudov\textsuperscript{1,\,2,\,5,\,$\dagger$} \quad
  Matija Medvidovi\'c\textsuperscript{4,\,$\dagger$}} \\[0.4em]
  \textsuperscript{1}Mila - Quebec AI Institute \quad
  \textsuperscript{2}Universit\'e de Montr\'eal \quad
  \textsuperscript{3}Princeton University \\[0.1em]
  \textsuperscript{4}ETH Zurich \quad
  \textsuperscript{5}Institut Courtois \quad
  \textsuperscript{$\dagger$}Equal supervision
}

\begin{document}

\maketitle
\etocdepthtag.toc{mainmatter}
\begin{abstract}
    Density functional theory (DFT) strikes a practical balance between accuracy and computational cost in many problems of computational chemistry and materials science. 
    However, many DFT calculations are limited by fixed atom-centered basis sets, which dictate how accuracy and cost scale with system size. 
    We propose \emph{Gaussian Splatting for Density Functional Theory} (GS-DFT), which represents molecular orbitals as a cloud of Gaussians whose positions, shapes, and mixing coefficients are optimized jointly by gradient descent to minimize the energy without training data. 
    Conceptually, GS-DFT is 3D Gaussian splatting with the renderer replaced by quantum mechanics.
    We introduce two key solver components: adaptive density fitting with screening for efficient evaluation of two-electron integrals, and a regularized differentiable orthogonalization of the molecular orbitals.
    Empirically, the optimized basis reaches the accuracy of the largest conventional basis sets with a fraction of the parameters, converging systematically in energy, density, and nuclear forces.
    At equal parameter count, it captures the stretched-bond and anion physics that fixed bases only recover with specialized basis augmentation.
    The resulting solver exhibits quadratic peak memory scaling in the cloud size, allowing us to simulate systems of up to $2,742$ atoms ($10,406$ electrons) without any modifications at triple-zeta scale using a single four-GPU node.
\end{abstract}

\section{Introduction} \label{sec:intro}

Density functional theory (DFT) offers reliable and practical approximations of many-electron quantum properties from first principles. 
With a uniquely favorable balance between accuracy and cost, it has matured as a computational pipeline, delivering new scientific insight \citep{Burke2012Perspective, jonesDensityFunctionalTheory2015} at scale. 
Many downstream areas of computational science see DFT as the only interface with quantum phenomena. 
It is the foundation of modern high-throughput material screening \citep{jainComputationalPredictionsEnergy2016a, curtaroloHighthroughputHighwayComputational2013, fiedlerDeepDiveMachine2022}, and data generation for machine-learned force fields \citep{unkeMachineLearningForce2021, jacobsPracticalGuideMachine2025a, hollingsworthMolecularDynamicsSimulation2018} in molecular dynamics. 
At the industrial scale, DFT is used in inverse design settings that require frequent evaluations of quantum mechanical properties \citep{butlerMachineLearningMolecular2018}. 
Routine applications include the conformational energy ranking of drug candidates \citep{niaziQuantumMechanicsDrug2025, kairysBindingAffinityDrug2019, guanApplicationDensityFunctional2025}, characterizing defects and interfaces in semiconductor chip design \citep{tranAccurateBandGaps2009, kimImprovementResistiveMemory2006, vandewalleFirstprinciplesCalculationsDefects2004}, and mapping reaction pathways or identifying active sites in computational catalysis \citep{nandyComputationalDiscoveryTransitionmetal2021, sameeraTransitionMetalCatalysis2012, liaoDensityFunctionalTheory2022}.
Therefore, all of these applications inherit DFT's biases and are limited by its data availability.


In practice, calculations in quantum chemistry are bottlenecked by three factors.
The first is scalability.
Electronic orbitals are expanded in a fixed set of $B$ basis functions, traditionally chosen as approximations of atomic orbitals (AO).
The memory cost of the traditional self-consistent field algorithm retains substantial intermediate-storage and tensor contraction costs \citep{linMathematicalIntroductionElectronic2019, lehtolaOverviewSelfConsistentField2020}, making numerical mitigation an area of active research \citep{beckePerspectiveFiftyYears2014}.
The second is physical fidelity.
Accuracy is bounded by how well those predetermined $B$ atomic orbitals can represent the true molecular orbitals.
Standard atom-centered basis sets have well-catalogued limitations \citep{lehtolaReviewNonrelativisticFully2019, jensenAtomicOrbitalBasis2013, klopperGaussianBasisSets1986} in representing known physics and chemistry without manual fine-tuning.
The third is the accuracy of the resulting observables.
Orbital basis errors are reflected and often amplified by output quantities, such as forces and densities.
Forces converge more slowly with basis size than energies and require additional correction terms \citep{Pulay1969Forces} for atom-centered bases.
The fixed-basis parameterization of electronic orbitals contributes to all three issues.
These limitations motivate a solver that is physically constrained by construction, amortizable across related solves, and naturally supports GPU parallelization.

We propose parametrizing electronic orbitals with \textit{floating} anisotropic Gaussian functions.
Instead of solving the conventional DFT self-consistency cycle, we directly minimize the total DFT energy using automatic differentiation, yielding a parallelizable direct optimization procedure. 
Our approach is similar to 3D Gaussian splatting \citep{Kerbl2023GaussianSplatting} with the data-driven photometric loss replaced by variational energy optimization.
In addition, we propose a memory- and compute-efficient \emph{density fitting} scheme for evaluating the two-electron energy contributions, reducing the peak memory cost from $\mathcal{O}(M^4)$ to $\mathcal{O}(M^2)$ and the per-step tensor contraction cost to $\mathcal{O}(M^2)$ for $M$ Gaussians.
The method is summarized in \Cref{fig:visual-abstract}.
On molecules of increasing size (\cref{sec:experiments}), we demonstrate:
\begin{itemize}
    \item \textbf{State-of-the-art accuracy scaling.} 
    We analyze the accuracy scaling laws for GTOs and GS-DFT, showing $2.3\times$ better scaling for all systems. We further show that our results hold for any common exchange-correlation functional type and reproduce the Hartree-Fock ionization potentials.  
    \item \textbf{Correct physics without expert corrections.} In regimes where fixed bases have to be augmented by experts (e.g. anions and stretched bonds), GS-DFT reproduces or surpasses the best results of augmented baselines with the augmented budget.
    \item \textbf{Adaptive density fitting for frontier systems.} 
    The \textit{adaptive} density fitting reduces the memory cost to $\mathcal{O}(M^2)$, so the largest system in this work, \textit{protease Nsp1-alpha from the porcine reproductive and respiratory syndrome virus} wiht $\sim 10,000$ electrons, can be optimized on a single node, thus reduceing the memory footprint by a factor of $\sim 20$.
\end{itemize}


\begin{figure}[t]
\captionsetup{skip=4pt}
    \centering
    \includegraphics[width=\linewidth]{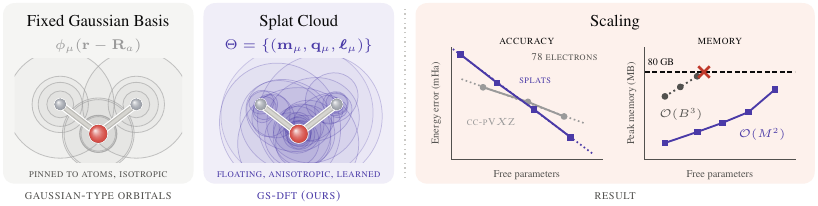}
    \caption{
        \textbf{Gaussian Splatting for Density Functional Theory (GS-DFT).}
        Molecular orbitals are expanded in a cloud of $M$ Gaussians (\emph{splats}) whose means and covariances $\Theta$ are optimized together with the coefficients $\mC$ by variational direct energy minimization.
        The rightmost panels report both scalings for alanine dipeptide ($N{=}78$ electrons). 
        The error against the extrapolated references falls faster per free parameter than for the traditional bases with $B$ functions, and peak memory grows as $\mathcal{O}(M^2)$ instead of the $\mathcal{O}(B^3)$ of most traditional codes.
    }
    \label{fig:visual-abstract}
\end{figure}
\section{Background} \label{sec:background}

\textbf{Quantum mechanics as a variational problem.}
Consider a molecule with $N$ electrons and $M$ nuclei with charges $Z_a$ at positions $\RR_a \in \mathbbm{R}^3$, indexed by $a = 1, \ldots , M$. 
Its exact quantum properties (the \emph{electronic structure}) are governed by the Hamiltonian
\begin{equation}
\label{eq:many-body-hamiltonian}
    \hat{H} \;=\; \sum_{i=1}^{N} \Big[ -\tfrac12 \laplacian_{\rr_i} + \vext(\rr_i) \Big] \;+\; \sum_{i<j} \frac{1}{|\rr_i - \rr_j|} \; ,
\end{equation}
an operator on antisymmetric $N$-electron wavefunctions $\Psi \in L^2(\mathbbm{R}^{3N})$, reflecting physical one- and two-electron interactions. 
For molecules at equilibrium, $\vext(\rr) = -\sum_a \flatfrac{Z_a}{|\rr - \RR_a|}$ is the Coulomb attraction to the nuclei. 
Hamiltonian eigenvalues represent the energies of electronic states. 
Absent external perturbations, electrons are found in the ground state $\Psi _0$, the lowest-energy antisymmetric eigenfunction \citep{sakuraiModernQuantumMechanics2017} of the Hamiltonian in Eq.~\ref{eq:many-body-hamiltonian}. 

Approximating many-electron states by direct minimization of $E[\Psi]$ is a frontier computational problem with prohibitive scaling caused by the exponential growth of the electronic state space \citep{medvidovicNeuralnetworkQuantumStates2024}. 
Kohn-Sham DFT (KS-DFT) \citep{kohnSelfConsistentEquationsIncluding1965} bypasses this obstacle by modeling a surrogate system of non-interacting electrons that yields the same charge density
\begin{equation}
\label{eq:wfn-to-density}
    \rho (\rr ) = N \int \dd[3]{\rr _2} \cdots \int \dd[3]{\rr _N} \left| \Psi (\rr, \rr _2, \ldots, \rr _N) \right| ^2 \quad \in \; L^1(\mathbbm{R}^{3})
\end{equation}
as the full theory given by Eq.~\ref{eq:many-body-hamiltonian}. 
KS-DFT is stated as an optimization problem over $N$ single-electron wave functions, or \emph{orbitals}, $\psi _k \in L^2(\mathbbm{R}^{3})$. 
Following Eq.~\ref{eq:wfn-to-density}, the density $\rho (\rr) = \sum _k |\psi _k (\rr)| ^2$ simplifies to a mixture of orbital Born distributions. 
The KS orbitals are then computed as minimizers of the energy functional $E[\psi]$. 
The total \emph{approximate} KS energy functional reads $E[\psi] = T_s [\psi] + E_\text{ext} [\rho] + E_H [\rho] + E_\xc [\psi] + E_\text{BG}$ where
\begin{equation}
\label{eq:energy-terms}
\begin{gathered}
    T_s [\psi] = \frac12 \sum_{k=1}^{N} \int \dd[3]{\rr} \, \left| \grad \psi _k (\rr) \right| ^2 \; , \quad
    E_\text{ext} [\rho] = \int \dd[3]{\rr} \rho(\rr) \, \vext (\rr) \; , \quad \\
    E_H [\rho] = \frac12 \int \dd[3]{\rr} \int \dd[3]{\rr'} \frac{\rho(\rr) \rho(\rr')}{|\rr - \rr'|}
    \; , \quad \text{and} \quad
    E_\text{BG} = \sum _{a < b} \frac{Z_a Z_b}{|\RR_a - \RR_b|}
\end{gathered}
\end{equation}
are the kinetic, external, Coulomb (\textit{Hartree}) and background contributions, respectively. 
The molecular geometry $\{\RR _a, Z_a \}$ is held fixed under the standard Born-Oppenheimer approximation. 
The final energy term is the exchange-correlation (XC) energy $E_\xc$. 
The XC energy models all many-electron quantum effects, despite commonly accounting for less than $1 \%$ of the total. 
Therefore, the XC functional has been the central amortized learning problem in DFT for decades \citep{burkePerspectiveDensityFunctional2012, beckePerspectiveFiftyYears2014}. 
With the total energy objective specified, DFT can be formulated as a constrained functional optimization problem
\begin{equation}
\label{eq:ksdft}
    \min _{\psi} \; E [\psi]
    \qquad \text{subject to} \qquad
    \int \dd[3]{\rr} \, \psi _k (\rr) \psi _l (\rr) = \delta _{k l} \; .
\end{equation}

\textbf{Gaussian-type orbitals.}
For molecular systems, orbitals are expanded in a fixed basis set of $B$ \emph{Gaussian-type atomic orbital} (GTO) functions $\chi _\mu$ as $\psi_k (\rr) = \sum_{\mu=1}^{B} \mC_{\mu k} \, \chi_\mu (\rr)$. Only the expansion coefficients $\mC$ are optimized while the Gaussian remains pinned to a given atom, with widths and other shape parameters determined once per element and stored in static tables \citep{Boys1950Gaussian, ditchfield1971self, dunning1989gaussian}.
Strictly linear combinations of atomic orbitals (LCAO) have remained the default choice because of analytical integration identities and efficient optimization. 
Firstly, every integral in Eq.~\ref{eq:energy-terms} has a closed form \citep{Helgaker2000MEST}.
Secondly, the linear structure is exploited by the self-consistent field (SCF) fixed-point iteration, which is used as an energy optimization proxy. 
SCF iterations repeatedly diagonalize a $B \times B$ effective Hamiltonian matrix \citep{Martin2004ElectronicStructure} at $\mathcal{O}(B^3)$ cost and combine the resulting density with the current best ground state estimate \citep{linMathematicalIntroductionElectronic2019}.


\textbf{Two-electron integrals.}
The LCAO formalism admits pre-computation of intermediate quantities, known as \emph{density fitting}~\citep{Dunlap1979Fitting, Dunlap2000Robust}, enabling fast on-demand two-electron integral estimation for a given system. 
In DFT, the Hartree energy integral $E_H$ defined in \cref{eq:energy-terms} accounts for the classical electrostatic energy of the electronic density $\rho$. 
In terms of the fixed basis $\chi$, the density is evaluated through its connection to KS orbitals,
\begin{equation}
\label{eq:density-matrix}
    \rho (\rr) = 2 \sum_k \left| \psi _k (\rr) \right| ^2 = \sum_{\mu \nu} \mP_{\mu \nu} \, \chi_\mu (\rr) \chi_\nu (\rr) \; ,
\end{equation}
with $\mP = 2 \bar{\mC} \bar{\mC}^\top$ labeled the \emph{density matrix}. 
Thus, the Hartree integral can be written out as
\begin{equation}
\label{eq:eri}
    E_H = \frac12 \sum_ {\mu \nu \lambda \sigma} \mP_{\mu \nu} \mP_{\lambda \sigma} \int \dd[3]{\rr} \int \dd[3]{\rr '} \frac{\chi_\mu (\rr) \chi_\nu (\rr) \chi_\lambda (\rr ') \chi_\sigma (\rr ')}{|\rr - \rr '|} \; .
\end{equation}
Despite all integrals over GTO products being analytically tractable, naive evaluation of the full electron-repulsion integral (ERI) tensor with $\mathcal{O}(B^4)$ entries in \cref{eq:eri} leads to an optimization that costs $\mathcal{O}(B^4)$ in FLOPs and memory.

Density fitting introduces an auxiliary basis as a set of functions $\{\varphi _a\}_{a=1}^{N_{\mathrm{aux}}}$ with $N_{\mathrm{aux}} \ll B^2$. 
For the Hartree integral in \cref{eq:eri}, the purpose of the auxiliary basis is to replace the quadratic expansion of the density in terms of the basis $\chi$ in \cref{eq:density-matrix} with a linear expansion $\rho (\rr) \approx \sum _{a=1} ^{N_{\mathrm{aux}}} \vd _a \, \varphi _a (\rr)$ in terms of the auxiliary basis. 
This replacement has the effect of bringing down the degree of the integrand in \cref{eq:eri} from quartic to quadratic in $\chi$. 
The method is called \emph{density fitting}~\citep{Dunlap1979Fitting, Dunlap2000Robust} because the coefficients $\vd$ in the auxiliary expansion are determined by a simple least-squares linear fit.
The error of this approximation is the Coulomb-metric projection residual of the full quadratic $\rho$ in \cref{eq:density-matrix} onto the space spanned by the auxiliary basis, second order in the fit error (\cref{prop:dunlap,prop:pairproduct}).
Auxiliary bases for GTOs are human expert-designed GTOs themselves. However, the choice of the auxiliary basis in traditional density fitting introduces a new source of error that has to be independently ablated for each DFT run.

\textbf{DFT and 3DGS share an optimization problem.}
In computer vision, 3D Gaussian Splatting (3DGS) \citep{Kerbl2023GaussianSplatting} is a state-of-the-art method for 3D scene reconstruction that represents a scene as a linear combination of $M$ anisotropic Gaussian primitives $g_{\theta_\mu}$, carrying linear appearance weights $c_\mu$ (opacity and color). 
It minimizes a photometric loss $\mathcal{L}(\cdot, \cdot)$ through a differentiable rendering (\textit{splatting}) operator $\mathcal{R}_v$ against reference views $I_v$ provided as data. 
The problem
\begin{equation}
\label{eq:3dgs}
    \min_{\Theta,\, \vc} \; \mathop{\mathbb{E}} _v  \; \mathcal{L}\left( \mathcal{R}_v \left[ \textstyle\sum_{\mu=1}^{M} c_\mu \, g_{\theta_\mu} \right], \; I_v \right)
\end{equation}
is solved jointly in the nonlinear primitive parameters $\Theta = \{\theta_\mu\}_{\mu=1}^{M}$ and the linear weights $\vc$ for each scene. 
Writing DFT over the same Gaussian primitives results in a problem similar to Eq.~\ref{eq:3dgs}, replacing the photometric loss with the total energy functional $E[\cdot]$ from \cref{eq:energy-terms}. 
When using GTOs, the conceptual difference between Equations~\ref{eq:3dgs}~and~\ref{eq:ksdft} is only in the constraint. 
Therefore, DFT is positioned to inherit the entire optimization toolkit that makes Eq.~\ref{eq:3dgs} work at scale. 
While \emph{splatting} in 3DGS refers to density projections onto two-dimensional view planes, DFT optimization can be understood as sequential projections onto the $N$-dimensional subspace of occupied electronic states.
\section{Related Work} \label{sec:related}

\textbf{Floating and adaptive bases.}
Floating basis functions began with the floating spherical Gaussian orbitals of \citet{Frost1967FSGO}, using one Gaussian per electron pair and made anisotropic by ellipsoidal by \citet{Vescelius1974Ellipsoidal, Cohen1976Ellipsoidal}.
Later on, \citet{Pederson1988FloatingGaussians} minimized the DFT energy over floating Gaussians with Pulay-free forces.
However, these approaches remained inaccurate and unscalable due to the lack of a suitable density fitting scheme.
Alternatively, materials simulations pursue adaptivity through refinable discretizations instead, finite elements \citep{Schauer2013JCPFiniteElements}, wavelets and multiwavelets \citep{Genovese2008Wavelets, Jensen2017Elephant}, and gausslets \citep{White2017Gausslets}, which refine space where needed but multiply degrees of freedom rather than moving them.
Most recently, ELECTRA \citep{Elsborg2025ELECTRA} trains an equivariant network on reference data to predict a molecule's density as a mixture of floating Gaussians that warm-starts a conventional solver, and \citet{medvidović2026multipolesplatsoptimizedinverted} optimize a superposition of multipole Gaussian splats to compute the optimized effective potential and solve the inverse Kohn-Sham problem.

\textbf{Machine learning for DFT.}
Learned components cover the whole pipeline: interatomic potentials replace the calculation \citep{BehlerParrinello2007NNP, Schutt2017SchNet, Batzner2022NequIP, Batatia2022MACE}, learned exchange-correlation functionals improve the energy calculation \citep{snyderFindingDensityFunctionals2012, neuralxc, chen2020deepks, kirkpatrick2021pushing}, and neural networks predict densities, Hamiltonians, and properties in a fixed atomic-orbital basis \citep{Brockherde2017HKMap, Qiao2020OrbNet, ungerPhiSNet2021, hazraPredictingDensityMatrix2024}.
Closest to us methodologically is differentiable, direct-optimization DFT: automatic differentiation through Hartree-Fock exposing basis exponents \citep{TamayoMendoza2018AutodiffHF}, DQC \citep{KasimVinko2022DQC} and D4FT \citep{liD4FTDeepLearning2023}, direct optimization with self-diagonalization in solids \citep{Li2024Diagonalization}, and amortized networks that propose orbitals \citep{hassanSelfRefiningTrainingAmortized2025}; all of them optimize within, or lightly relax, an atom-centered basis.
Instead, we target the basis itself, so every method above composes with GS-DFT.

\section{Method} \label{sec:method}

We introduce Gaussian Splatting for Density Functional Theory (GS-DFT), a DFT solver that represents the occupied Kohn-Sham orbitals with a learnable cloud of Gaussian primitives. 
To efficiently and stably compute the two-electron terms on the fly, we introduce \emph{adaptive} density fitting and a regularized eigensolve to maintain the fundamental orthogonality constraint. 
The total energy functional is optimized using direct variational energy minimization. 
\cref{fig:ablations} shows how each ingredient influences GS-DFT scaling compared to traditional solvers. 

\paragraph{Floating orbitals as 3D Gaussian Splats.}
We replace the conventional fixed atom-centered basis with a learnable set of $M$ normalized anisotropic Gaussian primitives, or \emph{splats},
\begin{equation}
\label{eq:splat}
    g _\theta (\rr) = \det(\tfrac{\mA}{\pi})^{\nicefrac14} \; \exp \left\{ {-\tfrac12} (\rr - \vm)^\top \mA \, (\rr - \vm) \right\},
\end{equation}
with center $\vm\in \mathbb R^3$ and symmetric positive-definite precision matrix $\mA \in \mathbb R^{3\times 3}$. Following 3DGS \citep{Kerbl2023GaussianSplatting}, we parametrize $\mA$ by its eigendecomposition $\mA = \mU(\vq) \diag(e^{\vl}) \mU(\vq)^\top$, with log-eigenvalues $\vl \in \mathbbm{R}^3$ and a unit quaternion $\vq$ whose rotation $\mU(\vq) \in SO(3)$ orients the principal axes.
Every $\theta = \{\vm, \vl, \vq\}$ therefore defines a valid splat with nine free parameters. 
We index splats with $\mu = 1, \ldots , M$ in arbitrary order, denoting the set of all geometric parameters $\Theta = \{ \theta _1, \ldots, \theta _M \}$ and using $g_\mu = g_{\theta _\mu}$ as a shorthand.
The $N_\text{occ}$ occupied orbitals are obtained by linearly mixing these shared primitives, $\psi _k (\rr) = \sum _\mu \rmC _{\mu k} \, g_\mu (\rr)$, with the coefficient matrix $\rmC \in \mathbbm{R}^{M \times N_\text{occ}}$. 
Because every splat is positive, the nodal structure of the orbitals comes from nearby splats with coefficients of opposite sign (Propositions \ref{prop:gto-limit} and \ref{prop:signed}). Proposition \ref{prop:equivariance} shows that this parametrization is equivariant under $O(3)$.

\begin{figure}[t]
\captionsetup{skip=4pt}
    \centering
    \includegraphics[width=\linewidth]{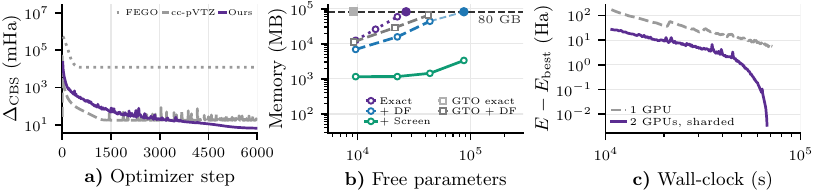}
    \caption{
        \textbf{Components of Gaussian Splatting for Density Functional Theory (GS-DFT).}
        \textit{(a) Parametrization of \cref{eq:splat}}, and the [Floating Elliptical Gaussian Orbitals, FSGO] \citep{Vescelius1974Ellipsoidal} optimization trajectories against converged GS-DFT of the same function count.
        \textit{(b) Adaptive density fitting of \cref{eq:adaptive-density-fitting}} with screened, chunked evaluation, plotted as peak memory against free parameters. 
        \textit{(c) The sharded scaling}, with respect to the best energy, against wall-clock time, on chain of 45 alanine residues ($M = 4299$).
    }
    \label{fig:ablations}
\end{figure}

\paragraph{Orthogonality constraint.} 
Because the Gaussian primitives in \cref{eq:splat} are not generally orthogonal as required by \cref{eq:ksdft}, we directly evaluate their overlap matrix $S(\Theta) \in \mathbbm{R}^{M \times M}$ as
\begin{equation}
\label{eq:ovlp}
    \rmS _{\mu \nu} (\Theta)  = \int \dd[3]{\rr} \, g_\mu (\rr) \, g_\nu (\rr)
    = \sqrt{\tfrac{\det \sqrt{\mA_\mu \mA_\nu}}{\det(\tfrac12 (\mA_\mu + \mA_\nu))}} \; \mK _{\mu \nu}
\end{equation}
with the matrix $\mK$ defined as
\begin{equation}
\label{eq:malahanobis}
    \mK_{\mu \nu} = \exp \left\{ -\tfrac12 (\vm_\mu - \vm_\nu ) ^\top \mA_\mu (\mA_\mu  + \mA_\nu) ^{-1} \mA_\nu (\vm_\mu - \vm_\nu) \right\} \; ,
\end{equation}
so that the orthonormality constraint from \cref{eq:ksdft} reduces to $\rmC^\top \rmS(\Theta) \rmC = I$. 
We explicitly enforce that constraint in the forward pass via L\"owdin orthonormalization,
\begin{equation}
\label{eq:lowdin}
    \bar{\mC} = \mC\, \mG^{\nicefrac{-1}{2}}
    \qquad \text{with} \qquad
    \mG = \mC^\top \mS(\Theta) \, \mC \, ,
\end{equation}
which satisfies $\bar{\mC}^\top \mS \bar{\mC} = I$ identically and is differentiable in both arguments. 
One backward pass propagates the orthonormality constraint into the basis gradients. 
We are left with an unconstrained minimization over $\{ \mC, \vm, \vl, \vq \}$, resembling \cref{eq:3dgs}.
We compute the L\"owdin retraction through a \emph{regularized eigendecomposition} $\mG = \mU \, \diag(\lambda) \, \mU ^\top$ as
\begin{equation}
\label{eq:regeigh}
    \mG^{\nicefrac{-1}{2}} = \mU \, \diag \left( f(\lambda) \right)\, \mU^\top,
    \qquad \text{where} \qquad
    f(\lambda) = \max \left(\lambda,\; \varepsilon \, \lambda_{\max} \right)^{\nicefrac{-1}{2}} \; .
\end{equation}
Bare gradients of an eigenvalue decomposition are proportional to inverse pairwise eigenvalue differences $(\lambda _i - \lambda _j)^{-1}$, causing numerical instabilities for degenerate matrices. 
Because of that, we implement a regularized backward pass in which the eigenvalue differences $|\lambda_i - \lambda_j| \leq \delta \lambda_{\max}$ are replaced by $\tfrac12 ( f'(\lambda_i) + f'(\lambda_j) )$. 
Pairs below the floor receive no gradient. 
Away from degeneracies, the regularization has no effect at convergence.

\begin{algorithm}[b]
\caption{Gaussian Splatting for Density Functional Theory (GS-DFT).}
\label{alg:gsdft}
\begin{algorithmic}[1]
\Require Nuclei $\{\RR_a, Z_a\}$; cloud size $M$; steps $T$; refresh period $K$; threshold $\tau$; functional $E_\xc$
\State $\Theta \gets \mathrm{ChemInit}(\RR, Z)$;\quad $\mC \gets \text{minAO}(\Theta)$ \Comment{\cref{app:recipe}}
\For{$r = 1, \dots, \lceil T/K \rceil$}
    \State $\mathcal{P} \gets \mathrm{stopgrad}\big[\{(\mu , \nu) : |S_{\mu \nu}| > \tau\}\big]$;\quad $\{h_i\}_{i=1}^{N_{\mathrm{aux}}} \gets \mathrm{stopgrad}\big[\mathcal{A}(\Theta)\big]$ \Comment{$\mathcal{O}(M N_{\text{aux}})$}
    \State $\mL \gets \mathrm{stopgrad}\big[\mathrm{chol}(\mQ + \lambda \mathbbm{1})\big]$ \Comment{$\mathcal{O}(N_{\mathrm{aux}}^3 / K)$}
    \For{$t = 1, \dots, K$}
        \State $\mP \gets 2\, \mC \, \mathrm{regeig}\big(\mC^\top \mS(\Theta)\, \mC\big) \mC^\top$ \Comment{$\mathcal{O}(M^2 N_{\mathrm{occ}} + N_{\mathrm{occ}}^3)$}
        \State $\vd \gets \mL^{-\top} \mL^{-1} \, \vt(\Theta, \mP)$ \Comment{$\mathcal{O}(M k N_{\mathrm{aux}} + N_{\mathrm{aux}}^2)$}
        \State $E \gets T_s[\mP] + E_{\mathrm{ext}}[\mP] + \tilde{E}_H[\vd] + E_\xc[\mP]$ \Comment{$\mathcal{O}(N_g M N_{\mathrm{occ}})$}
        \State $(\Theta, \mC) \gets \mathrm{Adam}\big((\Theta, \mC),\; \nabla_{\Theta,\mC}\, E\big)$
    \EndFor
\EndFor
\State \Return $(\Theta, \mC)$;\quad $\mF_a = -\,\partial E / \partial \RR_a$ \Comment{\cref{prop:forces}}
\end{algorithmic}
\end{algorithm}

\paragraph{Adaptive density fitting.}
\rev{Efficient evaluation of the variational energy objective requires fast evaluation of two-electron integrals.}
The Hartree energy of \cref{eq:eri} couples every pair of splat products, so evaluating it exactly costs $O(M^4)$.
For Gaussian splats, we extend the static density fitting protocol described in \cref{sec:background} with an adaptive auxiliary basis built from the splats themselves.
Because the density in \cref{eq:density-matrix} is a weighted sum of the splat pair-products, we exploit the fact that each product is a new Gaussian by sub-selecting $N_\text{aux} \ll M^2$ of these products to be our auxiliary basis.
Specifically, a product of Gaussian splats $g_\mu$ and $g_\nu$ is an \emph{auxiliary} Gaussian $h_{\mu\nu}$ defined by 
\begin{equation}
    g_\mu (\rr) \, g_\nu (\rr) = \mK _{\mu \nu} \, h_{\mu \nu} (\rr),
\end{equation}
where $h_{\mu \nu}$ has precision $\mA_{\mu \nu} = \mA_\mu + \mA_\nu$ and center $\vm_{\mu \nu} = (\mA_\mu + \mA_\nu)^{-1}(\mA_\mu \vm_\mu + \mA_\nu \vm_\nu).$
The prefactor $\mK _{\mu \nu}$ is identical to the one in \cref{eq:malahanobis} and directly determines which products matter in the expansion in \cref{eq:density-matrix}. 
As a result, keeping only the pairs with $| \mS_{\mu \nu} | > \tau$ is a consistent way of auxiliary basis \emph{screening}, controlled by the hyperparameter $\tau > 0$. 
For $\tau = 0$, we retain all splat pairs and return to evaluating all $\mathcal{O}(M^4)$ terms in \cref{eq:eri}. 
This screening method is related to Schwarz screening \citep{Haser1989DirectSCF}, which is commonly used in large DFT runs to exploit locality. 
For GS-DFT, it falls out naturally. 
We find that each splat has $k = \mathcal{O}(1)$ overlap neighbors, resulting in $N_\text{aux} = \mathcal{O}(M)$ auxiliary Gaussians in practical calculations.

After specifying the auxiliary basis $h_a$, it is frozen for $K$ optimization steps and detached from the computation graph. 
At each of these $K$ steps, the density is expanded as $\rho (\rr) = \sum _{a=1} ^{N_\text{aux}} \vd _a h_a (\rr)$ using the least-squares linear fit onto the density form in \cref{eq:density-matrix}, where the auxiliary index $a=(\mu, \nu)$ runs over splat pairs. 
The fit is analytic, resulting in $\vd = \mQ^{-1} \vt$ with closed-form integrals
\begin{equation}
\label{eq:adaptive-density-fitting}
    \mQ _{a b} = \int \dd[3]{\rr} \int \dd[3]{\rr'} \frac{h_a (\rr) h_b (\rr')}{|\rr - \rr'|}
    \quad \text{and} \quad
    \vt _a = \sum _{\mu \nu} \mP _{\mu \nu} \int \dd[3]{\rr} \int \dd[3]{\rr'} \frac{g_\mu (\rr) g_\nu (\rr) h_a (\rr')}{|\rr - \rr'|}
\end{equation}
available because $h_a$ are simple Gaussians themselves. After $K$ gradient updates to primitive Gaussians $g_\mu$, the auxiliary basis is refreshed and the procedure repeats.

\paragraph{Natively differentiable simulation.} 
Every step in the forward map from $(\mC, \Theta)$ to the energy $E$ is differentiable. 
This allows us to replace the SCF algorithm with end-to-end optimization by gradient descent or any of its modifications, e.g., Adam \citep{kingmaAdamMethodStochastic2015}. 
The splat parametrization being equivariant under O(3) (\cref{prop:equivariance}), the regularized eigensolve keeping gradients finite at degenerate eigenvalues (\cref{prop:regeigh}), and the adaptive density fitting bounding the fitting error as the cloud moves (\cref{prop:pairproduct}) all keep this optimization well-behaved. \Cref{alg:gsdft} gives the full training loop. The auxiliary metric of \cref{eq:adaptive-density-fitting} is factorized once per refresh and reused through $K$ triangular solves, making the dominant cubic operation in the method scale as $\mathcal{O}(N_{\mathrm{aux}}^3/K)$ amortized.
We exploit the native differentiability of our simulation in computing forces after the given GS-DFT run has converged. 
Because no basis function is anchored to an atom, the converged atomic force can be evaluated by direct automatic differentiation. 
Therefore, the Pulay correction of atom-centered bases \citep{Pulay1969Forces} is absent, resulting in a cheaper operation.
\section{Experiments} \label{sec:experiments}
\vspace{-0.7\baselineskip}
In this section, we show three advantages of GS-DFT: accuracy and memory scaling for many tested XC functionals and molecules, diffuse physics representation in stretched molecular geometries, and access to frontier systems using common hardware.

\begin{figure}[t]
\captionsetup{skip=4pt}
\centering
\includegraphics[width=\linewidth]{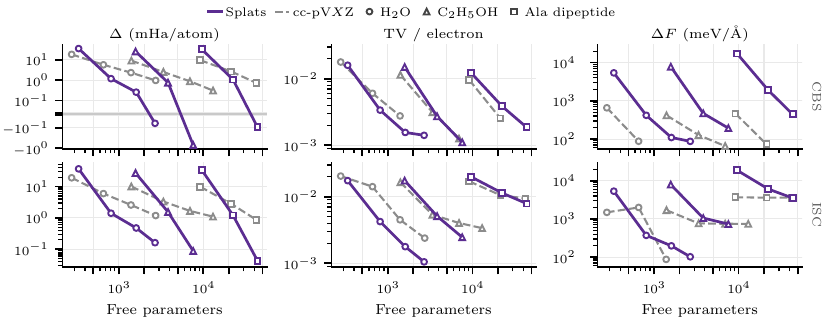}
\caption{
    \textbf{Observables against free parameters.}
    Columns are the energy residual per atom, the total variation of the density per electron, and the maximum deviation of the nuclear forces. \textbf{Top row:} measured against the CBS limit. \textbf{Bottom row:} measured against the ISC limit.
}
\label{fig:observables}
\vspace{-0.5\baselineskip}
\end{figure}

\subsection{Accuracy and functional coverage}
\label{sec:exp-observables}
\vspace{-0.5\baselineskip}

\paragraph{Accuracy.}
After optimizing the DFT energy loss, the orbitals are used to make predictions about the quantum-mechanical properties of the molecule. 
At $\grad _{\Theta} E \approx 0$, the energy errors are second-order in small parameter shifts. 
However, all other downstream observables are first-order sensitive to these small parameter variations (\cref{prop:forces}), requiring numerically precise optimization. 
There are two main quantities in this category: the electronic density $\rho$ itself and the atomic forces, given by the Hellmann-Feynman derivative $\mF _a = -\pdv*{E}{\RR_a}$ at convergence.
To show convergence towards a physical density, we optimize a splat cloud on three benchmark molecules of increasing complexity: water (\ce{H2O}), ethanol (\ce{C2H5OH}), and alanine dipeptide (\ce{C6H12N2O2}). 
We track the three key observables: energy, density, and nuclear forces (\cref{fig:observables}).

In \cref{fig:observables}, we show observables relative to extrapolated references as a function of the number of free parameters in GS-DFT, supplemented with GTO results at a matched parameter count.
Because DFT is solved as a data-free optimization, unbiased external solutions are difficult to obtain, especially at scale.
In \cref{fig:observables}, we use two types of references: the GTO complete basis set (CBS) limit and the infinite splat cloud (ISC) limit.
For any observable $O_P$ computed with a fixed budget of $P$ free parameters, the CBS and ISC limits are independently defined by $P \to \infty$.
In practice, we assume that $O_P - O_\infty \propto P^{-\alpha}$ and compute $O_\infty$ and the scaling exponent $\alpha$ by simple regression.
We find that GS-DFT consistently produces lower energies in the high-parameter regime in \cref{fig:observables}. All reported energies in this work are evaluated with full ERIs, eliminating density fitting errors. Density predictions show reliable convergence and similar or better quality than GTO references, as measured by the dimensionless total variation $\mathrm{TV} = \int \dd[3]{\rr} |\rho_P (\rr) - \rho _\infty (\rr) |$.
For force predictions, the CBS-ISC distinction is required because the absolute ordering of errors on \cref{fig:observables} may depend on the chosen baseline.
However, convergence rate $\alpha$ estimates remain similar.
A third, independently discretized and converged accurate method would resolve absolute force accuracy.
\rev{\Cref{app:observables-scaling} repeats this analysis on six further systems with 7 to 74 atoms.}

\vspace{-0.5\baselineskip}
\label{sec:xc-coverage}
\paragraph{Functional coverage.} 
The splat representation of KS orbitals is independent of the choice of the XC functional.
We test this on ethanol across four rungs of functional complexity: PBE~\citep{perdewGeneralizedGradientApproximation1996} (a generalized gradient approximation, GGA), r$^2$SCAN~\citep{furnessAccurateNumericallyEfficient2020} (a meta-GGA), the global hybrid B3LYP~\citep{beckeDensityfunctionalThermochemistryIII1993}, and the range-separated hybrid $\omega$B97M-V~\citep{mardirossianWB97MVCombinatoriallyOptimized2016}.
In \cref{tab:xc-coverage}, we compare a GTO reference, cc-pVTZ basis, to GS-DFT at matched function count, $N_{\mathrm{ao}}=M = 174$.
Energies achieved by GS-DFT sit closer to the CBS limit than the GTO reference for all four functionals. 
For the widely-used PBE and r$^2$SCAN functionals, we achieve optimized energy values in the range of the cc-pV5Z GTO basis, using $29\%$ of its functions. 
In addition, reliable convergence and predictions across a wide variety of XC approximations suggest that GS-DFT predicts physical densities instead of exploiting numerical pathologies in the non-linear optimization with adaptive density fitting.

\begin{table}[t]
\vspace{-0.5\baselineskip}
\captionsetup{skip=4pt}
\centering
\footnotesize
\setlength{\tabcolsep}{4pt}
\begin{minipage}[t]{0.366\linewidth}
\centering
\vspace{0pt}
\caption{\textbf{Functional Coverage.} Ethanol's $\mathbf{\Delta}_{\text{CBS}}$ with $M = n_{\mathrm{ao}} = 174$. $\% \mathbf{M}$ denotes the percentage of $\mathbf M$ used by our method.}
\label{tab:xc-coverage}
\vspace{-0.2\baselineskip}
\resizebox{\linewidth}{!}{
\begin{tabular}{lrrr}
\toprule
& \multicolumn{2}{c}{$\mathbf{\Delta}_{\text{CBS}}\downarrow$} & \\
\cmidrule(lr){2-3}
\textbf{Functional} & GS-DFT & cc-pVTZ & \textbf{\%M} \\
\midrule
PBE                 & $\mathbf{2.3}$  & $22.8$ & $<29$\\
B3LYP               & $\mathbf{16.9}$ & $21.3$ & $77$\\
r$^2$SCAN           & $\mathbf{1.8}$  & $19.9$ & $<29$\\
$\omega$B97M-V      & $\mathbf{10.2}$ & $19.8$ & $61$\\
\bottomrule
\end{tabular}
}
\end{minipage}\hfill
\begin{minipage}[t]{0.61\linewidth}
\centering
\vspace{0pt}
\caption{
\textbf{Ionization potential at matched function count.}
Under Hartree--Fock, $-\varepsilon_{\text{HOMO}}$ \emph{is} the ionization potential, in eV. GS-DFT uses $\textbf{M} = n_{\mathrm{ao}}$ of def2-QZVPP; the references are the CCSD(T) and experimental values of GW100.
}
\label{tab:gaps}
\vspace{-0.2\baselineskip}
\resizebox{\linewidth}{!}{
\begin{tabular}{lrrrrr}
\toprule
& & \multicolumn{2}{c}{$-\bm{\varepsilon}_{\text{HOMO}}$} & \multicolumn{2}{c}{\textbf{Reference (GW100)}} \\
\cmidrule(lr){3-4}\cmidrule(lr){5-6}
\textbf{System} & \textbf{M}$\:= n_{\mathrm{ao}}$ & GS-DFT & def2-QZVPP & CCSD(T) & Exp. \\
\midrule
\ce{Li2}    & $70$  & $4.950$  & $4.950$  & $5.266$  & $4.73$  \\
\ce{H2O}    & $117$ & $13.878$ & $13.870$ & $12.565$ & $12.62$ \\
\ce{C2H5OH} & $351$ & $12.042$ & $12.040$ & $10.685$ & $10.64$ \\
\ce{C6H6}   & $522$ & $9.148$  & $9.149$  & $9.29$   & $9.24$  \\
\bottomrule
\end{tabular}
}
\end{minipage}
\vspace{-0.5\baselineskip}
\end{table}

\vspace{-0.5\baselineskip}
\paragraph{Ionization Potential.}
The negative of the highest occupied orbital energy equals the ionization potential (IP) in the Hartree-Fock (HF) approximation \citep{koopmansUeberZuordnungWellenfunktionen1934}.
We treat HF as a special case of DFT, with an energy functional in \cref{eq:energy-terms} set to include exchange and ignore correlation.
We evaluate the largest orbital eigenvalue $\bm{\varepsilon}_{\text{HOMO}}$ against the GW100 \citep{vanSetten2015GW100}, a high-precision benchmark of one hundred closed-shell molecules whose IPs are computed with coupled cluster with singles, doubles, perturbative triples (CCSD(T)). 
We further add the measured experimental values to compare against \citep{maggio2017GW100PlaneWave}.
At matched function-count, GS-DFT reproduces the GTO references across all systems, as shown in \cref{tab:gaps}.
The remaining error of both GS-DFT and GTOs with respect to CCSD(T) and experimental values is attributed to orbital relaxation and the absence of electronic correlation effects.
In addition to the energy, the density, and the forces, the IP serves as an independent numerical test against high-precision numerics and measurements.

\vspace{-0.2\baselineskip}
\subsection{Splats as a physical representation}
\label{sec:exp-fidelity}
\vspace{-0.5\baselineskip}

Narrow conventional basis sets anchored at nuclei can lack the support needed to describe density that extends far from them.
In bare anions, the electron density expands away from the nuclei because of the extra electron-electron repulsion.
Traditionally, this problem in anions and stretched bonds has been mitigated by \emph{augmenting} basis sets with extra diffuse fixed Gaussians, e.g., the aug-cc-pV$\zeta$Z families \citep{Kendall1992AugBasis}.
This augmentation needs to be triggered by human experts before a calculation is started, as GTOs cannot automatically expand their support during optimization.
Splats have adaptive support and do not require expert intervention to represent diffuse systems.
In this subsection, we test GS-DFT on anions with long density tails and stretched chemical bonds.

\begin{figure}[b]
\vspace{-0.5\baselineskip}
\captionsetup{skip=4pt}
\centering
\includegraphics[width=\linewidth]{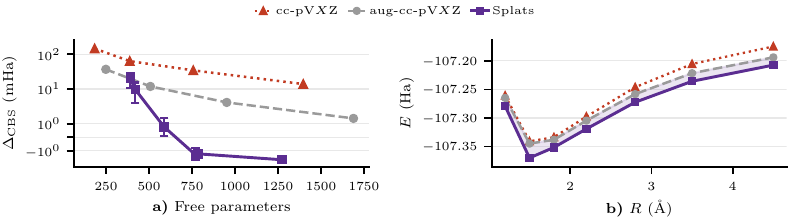}
\caption{
    \textbf{(a) Fluoride ion (\ce{F-}) basis-set convergence.} $\mathbf{\Delta}_{\text{CBS}} = E - E_{\text{CBS}}$ is the residual from the CBS limit in milli-Hartree (mHa).
    \textbf{(b) \ce{LiF} dissociation at matched count.}
    The margin over the unaugmented basis grows as the bond stretches, while the margin over the augmented basis stays flat.
}
\label{fig:anions}
\end{figure}

The left plot in \cref{fig:anions} shows how the splats converge to better energy values than the augmented baseline, cc-pVTZ. 
In GS-DFT computations in \cref{fig:anions}, we use the same number of functions as the unaugmented basis set.
For stretched bond geometries, electronic densities can maintain nonzero values in the interstitial region. 
In those cases, GTOs cannot extend their support to faithfully represent relevant delocalized orbitals.
In the right plot in \cref{fig:anions}, we show this behavior by stretching lithium fluoride (\ce{LiF}). 
We show that splats with adaptive support achieve lower energies compared to the GTO calculations with both cc-pVTZ and the manually augmented aug-cc-pVTZ variant, while using the unaugmented number of basis functions. 
Additionally, direct energy optimization in GS-DFT has shown robust convergence properties for all considered geometries, while SCF has displayed well-known convergence problems for highly deformed molecules.
See \cref{app:dissociation} for details.

\begin{figure}[t]
\captionsetup{skip=4pt}
\centering
\includegraphics[width=\linewidth]{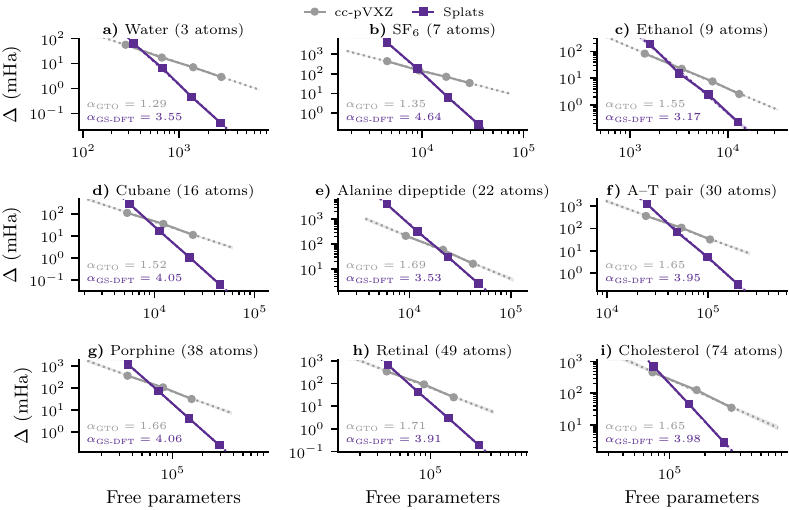}
\caption{\rev{
    \textbf{Accuracy scaling law for GS-DFT}. GTO rungs against the extrapolated complete basis set (CBS) energy, and the splat clouds against the infinite splat cloud (ISC) energy. Limits are obtained by fitting $\Delta = E_P - E_\infty \propto P^{-\alpha}$ against four runs with $P$ parameters.
}}
\label{fig:isoparams}
\vspace{-0.6\baselineskip} 
\end{figure}

\begin{figure}[b]
\captionsetup{skip=4pt}
\centering
\includegraphics[width=\linewidth]{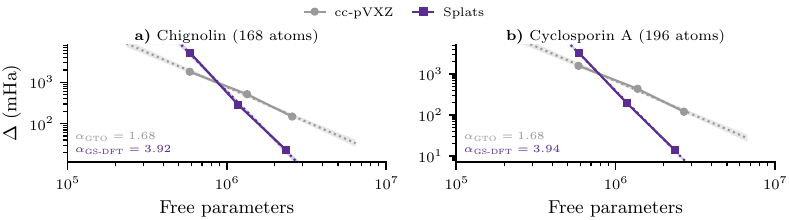}
\caption{\rev{\textbf{Accuracy scaling law for the intermediate systems.} As in \cref{fig:isoparams}, with three runs per ladder, so the splat clouds are also measured against the CBS energy.}}
\label{fig:isoparams-intermediate}
\end{figure}

\subsection{Accuracy and memory scaling}
\label{sec:exp-scaling}
\vspace{-0.5\baselineskip}

Finally, we examine two notions of scaling for GS-DFT: how rapidly basis-set errors decrease with representation expressivity and how accuracy and the memory footprint grow with molecule size. 

\vspace{-0.7\baselineskip}
\paragraph{Accuracy scaling law.}
We study the scaling of the energy error with the number of free parameters $P$ in \cref{fig:isoparams}.
We define the energy error and its scaling power law as $\Delta = E - E_\infty \propto P^{-\alpha}$, independently for GTOs and splats.
Our analysis, in \rev{\cref{fig:isoparams,fig:isoparams-intermediate}}, indicates that GS-DFT can achieve $2.3 \times$ sharper scaling exponents on average.
Across all systems studied in this work, we observe similar scaling improvements\rev{ (\cref{app:size-scaling})}.
This observation supports the intuition that using splats as an \emph{adaptive} basis is more parameter efficient because scaling is determined by the number of basis functions, regardless of how each function is parametrized. 
Therefore, packing more expressivity into each function is an efficient way to gain accuracy with minimal scalability cost.

\begin{table}[t]
\vspace{-0.6\baselineskip}
\centering
\caption{
    \textbf{FMO systems.}
    Energy in Ha, memory in GB. \rev{One $4\times$H100 GPU node has 320~GB in total.}
}
\label{tab:size-ladder}
\vspace{-0.7\baselineskip}
\resizebox{\linewidth}{!}{
\begin{tabular}{lrrrrrrcccc}
\toprule
 & & & \multicolumn{2}{c}{\textbf{Peak Mem.}} & & & \multicolumn{2}{c}{\textbf{Fits GPU node}} & \multicolumn{2}{c}{\textbf{Simple forces}} \\
\cmidrule(lr){4-5}\cmidrule(lr){8-9}\cmidrule(lr){10-11}
\textbf{System} & \textbf{Atoms} & \textbf{M} & GS-DFT & GTO & $\mathbf{E}_{\text{GS-DFT}}$ & $\mathbf{E}_{\text{HF}}$ & GS-DFT & GTO & GS-DFT & GTO \\
\midrule
Trp-cage            & 304   & 6,720  & \rev{$9.5$}  & $640^\ast$   & $-7{,}478.3$  & $-7{,}439.0$  & \cmark & \xmark & \cmark & \xmark \\
Defensin            & 465   & 10,390 & \rev{$16.7$}  & $1{,}060^\ast$ & $-13{,}685.4$ & $-13{,}624.0$ & \cmark & \xmark & \cmark & \xmark \\
Insulin             & 947   & 21,026 & \rev{$40.1$}  & $2{,}400^\ast$ & $-25{,}523.2$ & $-25{,}404.2$ & \cmark & \xmark & \cmark & \xmark \\
$\beta$-spectrin CH & 1,755 & 38,794 & \rev{$104.2$} & $4{,}850^\ast$ & $-44{,}255.4$ & $-44{,}045.7$ & \cmark & \xmark & \cmark & \xmark \\
3IFU (chain A)      & 2,742 & 60,308 & \rev{$229.0$} & $8{,}070^\ast$ & $-69{,}823.2$ & $-69{,}631.4$ & \cmark & \xmark & \cmark & \xmark \\
\bottomrule
\end{tabular}}
\par\vspace{0.2em}
{\footnotesize$^\ast$Density-fitted cc-pVTZ, extrapolated from measured runs.}
\vspace{-0.5\baselineskip}
\end{table}

\vspace{-0.7\baselineskip}
\paragraph{Memory scaling.} 
Scaling DFT beyond hundreds of atoms requires careful optimization, parallelization, and high-performance computing (HPC) using linear-scaling methods \citep{Goedecker1999LinearScaling, Bowler2012LinearScaling, Kitaura1999FMO} that trade off accuracy for scale \citep{Prodan2005Nearsightedness}. 
To show memory scalability, we run five systems from the Fragment Molecular Orbital (FMO) database \citep{takayaFMODBWorldFirst2021, watanabeAutoFMOProtocol2019}, and compare our GS-DFT results with their references, calculated with the linear-scaling ABINIT-MP package \citep{tanakaElectronCorrelatedFMO2014} implementing Hartree-Fock (HF) with 6-31g(d) basis using state-of-the-art HPC resources \citep{doiLargeScaleFMOMP2Spike2025}.
In contrast, we tackle the same large systems with the PBE functional, using a fraction of the memory estimated for the density-fitted cc-pVTZ GTO implementation and inter-node communication, see \cref{tab:size-ladder}.
GS-DFT's favorable memory scaling and data-parallel direct minimization allow us to take advantage of modern hardware and codes used in distributed AI training.
\rev{\Cref{app:parallel-efficiency} reports the speedup and the memory reduction over the GPUs of one node.}

\section{Conclusion}
\label{sec:conclusion}
\vspace{-0.5\baselineskip}
Every DFT calculation is currently limited by its fixed basis set.
We propose GS-DFT, expanding the orbitals in a cloud of floating Gaussians and optimizing their parameters by gradient descent on the variational energy loss.
This enables an \emph{adaptive} representation that achieves lower energies, delivers correct observables, and automatically captures diffuse physics.
Its memory scaling brings frontier quantum systems within reach of modern hardware accelerators.
However, GS-DFT requires more optimization steps and thus more time to converge.
We see this constraint as an opportunity for future research, e.g., physics-informed natural gradient optimization.
Numerical cost is also amortizable by recent methods at the interface of foundation models and 3DGS, taking advantage of a uniquely parameter- and memory-efficient representation of quantum physics offered by GS-DFT.

\ificlrfinal
\subsubsection*{Acknowledgments}

A.G-C. acknowledges support from the Secretaría de Ciencia, Humanidades, Tecnología e Innovación (SECIHTI) under a doctoral fellowship CVU 2014613.
C.Z. acknowledges support from the U.S. National Science Foundation under grant OAC 2118201.
M.S. was supported by the IVADO 2025 Postdoctoral Research Funding Program.
M.M. acknowledges support from a SwissAI grant from the Swiss National Supercomputing Centre (CSCS) under project ID a164.
This research was enabled in part by compute resources provided by Mila (\href{https://mila.quebec/en}{mila.quebec}) and the Digital Research Alliance of Canada (\href{https://www.alliancecan.ca/en}{alliancecan.ca}). A.G-C. and M.H. acknowledge the support of the Natural Sciences and Engineering Research Council of Canada (NSERC), funding reference number RGPIN-2026-05517. 
This project was undertaken thanks to funding from IVADO and the Canada First Research Excellence Fund.
\fi

\subsubsection*{Reproducibility Statement}
The code is available in \coderepokind{} at \coderepo.
It contains the implementation, the configuration and launch script of every experiment, the data every figure is built from, and every molecular geometry except the FMO proteins, which are available from FMODB \citep{takayaFMODBWorldFirst2021}.
The proofs of the theoretical results are in \cref{appendix:proofs}.
\Cref{appendix:implementation} gives the integrals, the training settings (\cref{tab:recipe}), the evaluation protocol, and the hardware and software versions.

\subsubsection*{AI Use Statement}
We used large language models to assist with writing and debugging code, with formatting the manuscript, and with finding relevant literature.
The authors designed the method and the experiments, checked the code, the results and every cited work, and take full responsibility for the content of the paper.

\subsubsection*{Impact Statement}
Gaussian Splatting for Density Functional Theory is a solver for \textit{ab-initio} molecule simulation.
Its consequences are those of better DFT tooling: efficient quantum-mechanical data generation with a memory footprint that keeps larger systems on common hardware.
We foresee no societal risks specific to this method beyond those of computational chemistry at large.

\bibliography{bibliography, bibliography-matija}
\bibliographystyle{iclr2027_conference}

\clearpage
\appendix

\begin{center}
{\Large\bfseries \papertitle\\[0.2em] (Supplementary Material)\par}
\end{center}
\vspace{0.5em}
\etocdepthtag.toc{appendix}
\etocsettagdepth{mainmatter}{none}
\etocsettagdepth{appendix}{subsection}
\etocsettocstyle{}{}
\tableofcontents
\clearpage

\section{Theoretical properties of GS-DFT}\label{appendix:proofs}

GS-DFT replaces a fixed, atom-centered basis with floating 3D Gaussians whose parameters the optimizer is free to move.
The extra degrees of freedom are what make the representation expressive, trading off the theoretical guarantees that come for free with an atom-centered basis.
This appendix restores them one at a time.
Specifically, we provide the theoretical support for the following questions:
\begin{itemize}
    \item \textbf{\cref{app:proofs-symmetry}}: Does a representation that is no longer attached to the atoms still respect the symmetries of the molecule?
    \item \textbf{\cref{app:proofs-forces}}: Are the forces obtained by differentiating the energy correct, and how wrong are they when the optimization stops early?
    \item \textbf{\cref{app:proofs-df}}: What does the adaptive density fit cost in accuracy, and can the optimizer exploit it?
    \item \textbf{\cref{app:proofs-regeigh}} What does the regularized eigensolve change about the final answer?
\end{itemize}
Each subsection explains why the question matters, introduces the objects needed to answer it, and closes with what the implications of its results.

\subsection{Symmetry of the representation}\label{app:proofs-symmetry}

The energy of a molecule does not depend on where it sits or how it is oriented.
For an atom-centered basis, this invariance holds automatically because every basis function is attached to a nucleus and moves with it.
A splat cloud is not attached to anything. 
Its centers, orientations, and scales are free parameters, so invariance is no longer automatic.
Energy invariance holds only if every rigid motion of the molecule can be absorbed by a transformation of the parameters that leaves the energy unchanged.
\Cref{prop:equivariance} shows that such a transformation exists, is explicit, and acts on the parameterization of \cref{eq:splat} as a group action.
The symmetry is, therefore, a property of the naturally equivariant model architecture, not something the optimizer has to learn.

Two facts about the parameterization are needed.
First, unit quaternions $\vq \in S^3$ form a group under the Hamilton product $\otimes$, and the map $\mU : S^3 \to SO(3)$ that sends a quaternion to its rotation matrix is a surjective homomorphism,
\begin{equation*}
    \mU(\vq_1 \otimes \vq_2) = \mU(\vq_1)\, \mU(\vq_2) \; ,
\end{equation*}
whose kernel is $\{\pm 1\}$.
The sign of $\vq$ is therefore a discrete gauge: $\vq$ and $-\vq$ define the same splat.
Second, every $(\vq, \vl) \in S^3 \times \mathbbm{R}^3$ yields a symmetric positive-definite precision matrix
\begin{equation*}
    \mA = \mU(\vq)\, \diag(e^{\vl})\, \mU(\vq)^\top \; ,
\end{equation*}
with the closed forms
\begin{align*}
    \mA^{-1} &= \mU \diag(e^{-\vl})\, \mU^\top \; , \\
    \det \mA &= e^{\sum_i \ell_i} \; , \\
    \lambda_{\min}(\mA) &= e^{\min_i \ell_i} \; .
\end{align*}
The optimization is therefore unconstrained, and the normalization in \cref{eq:splat} and the scale of every splat are available without a matrix decomposition.

\begin{proposition}[Equivariance of the splat parameterization]\label{prop:equivariance}
Let $\phi(\rr) = \mY \rr + \vb$ be a rigid motion with $\mY \in O(3)$ and $\vb \in \mathbbm{R}^3$.
Choose a unit quaternion $\vq_\mY$ with $\mU(\vq_\mY) = \mY$ if $\det \mY = 1$, and $\mU(\vq_\mY) = -\mY$ if $\det \mY = -1$.
Map each splat as
\begin{equation*}
    \theta_\mu = (\vm_\mu,\; \vq_\mu,\; \vl_\mu)
    \;\longmapsto\;
    \theta'_\mu = (\mY \vm_\mu + \vb,\; \vq_\mY \otimes \vq_\mu,\; \vl_\mu),
\end{equation*}
keep the coefficients $\mC$, and move the nuclei to $\phi(\RR_a)$.
Then, for every splat $\mu$ and every orbital $k$,
\begin{align*}
    g_{\theta'_\mu} \circ \phi &= g_{\theta_\mu} , \\
    \psi'_k \circ \phi &= \psi_k , \\
    \rho' \circ \phi &= \rho ,
\end{align*}
and every term of \cref{eq:energy-terms} takes the same value, as does the density-fitted Hartree energy of \cref{eq:adaptive-density-fitting}.
\end{proposition}

\begin{proof}
The homomorphism property gives
\begin{align*}
    \mU(\vq_\mY \otimes \vq_\mu) &= \mY\, \mU(\vq_\mu) && \text{if } \det \mY = 1, \\
    \mU(\vq_\mY \otimes \vq_\mu) &= -\mY\, \mU(\vq_\mu) && \text{if } \det \mY = -1.
\end{align*}
In both cases the sign cancels in the quadratic form, so
\begin{align*}
    \mA'_\mu &= \mY\, \mU(\vq_\mu)\, \diag(e^{\vl_\mu})\, \mU(\vq_\mu)^\top \mY^\top \\
             &= \mY \mA_\mu \mY^\top ,
\end{align*}
and $\det \mA'_\mu = \det \mA_\mu$ because $(\det \mY)^2 = 1$.
The normalization of \cref{eq:splat} is unchanged, and for every $\rr$,
\begin{align*}
    (\phi(\rr) - \vm'_\mu)^\top \mA'_\mu\, (\phi(\rr) - \vm'_\mu)
    &= (\rr - \vm_\mu)^\top \mY^\top \mY \mA_\mu \mY^\top \mY\, (\rr - \vm_\mu) \\
    &= (\rr - \vm_\mu)^\top \mA_\mu\, (\rr - \vm_\mu),
\end{align*}
so $g_{\theta'_\mu} \circ \phi = g_{\theta_\mu}$.
Because $|\det \mY| = 1$, the change of variables $\rr \mapsto \phi(\rr)$ preserves every integral, so the overlap matrix $\mS$ of \cref{eq:ovlp} is unchanged.
The occupied Gram matrix and the L\"owdin coefficients $\bar{\mC}$ of \cref{eq:lowdin} are then unchanged, and so is the density matrix $\mP$.
It follows that $\psi'_k \circ \phi = \psi_k$ and $\rho' \circ \phi = \rho$.

Each energy term is invariant for a separate reason.
For the kinetic energy,
\begin{equation*}
    \grad (\psi_k \circ \phi^{-1}) = \mY\, (\grad \psi_k) \circ \phi^{-1},
\end{equation*}
and $\mY$ preserves norms.
For the external energy,
\begin{equation*}
    |\phi(\rr) - \phi(\RR_a)| = |\rr - \RR_a| ,
\end{equation*}
so $v'_{\mathrm{ext}} \circ \phi = v_{\mathrm{ext}}$.
The Hartree and background terms depend only on the distances $|\rr - \rr'|$ and $|\RR_a - \RR_b|$, which $\phi$ preserves.
Semilocal exchange-correlation functionals depend on $\rho$ and on the rotational scalars $|\grad \rho|^2$, $\laplacian \rho$, and the kinetic-energy density $\tfrac12 \sum _k |\grad \psi _k |^2$.
Exact exchange (the Fock term in chemistry) is a Coulomb integral of orbital pair densities, so both are invariant.
Finally, each auxiliary function of the adaptive fit is a pair product $g_\mu g_\nu$, so it transforms like the splats, and the matrices $\mQ$ and $\vt$ of \cref{eq:adaptive-density-fitting} are unchanged.
\end{proof}

The proposition holds for the full orthogonal group $O(3)$.
A reflection needs no special treatment because a precision matrix is quadratic in the frame, so $\mY$ and $-\mY$ act identically.
The quaternion never has to represent an improper rotation.
The log-eigenvalues $\vl$ are frame invariants, so the scales of a cloud carry no information about its orientation.
Only the centers and the orientations transform, the latter by left multiplication,
\begin{align*}
    \vm &\mapsto \mY \vm + \vb , \\
    \vq &\mapsto \vq_\mY \otimes \vq .
\end{align*}

The result thus concerns the objective and not the optimizer.
A converged cloud for one geometry maps to a converged cloud for every rigidly moved copy with the same energy.

\subsection{Nuclear forces}\label{app:proofs-forces}

Nuclear forces are the negative gradient of the potential-energy surface
\begin{equation*}
    E^\star(\RR) = \min_\theta\, E(\RR, \theta),
\end{equation*}
the ground-state energy as a function of the nuclear coordinates.
They drive geometry optimization and molecular dynamics, and they are the labels on which machine-learned interatomic potentials are trained (\cref{sec:exp-observables}).
Because $E^\star$ is the value function of an inner optimization problem, computing its gradient is an instance of the bilevel differentiation familiar from hyperparameter optimization and implicit layers.
At an exact minimizer, the envelope theorem gives the answer.
The response of the parameters to the nuclei does not contribute, because the energy is stationary in them, and only the explicit dependence of $E$ on $\RR$ remains.

For an atom-centered basis, this explicit dependence is expensive.
Each basis function $\chi_\mu(\rr - \RR_{a(\mu)})$ moves with its nucleus, so $\RR$ enters every integral of \cref{eq:energy-terms} and not only the electron--nucleus attraction.
The envelope theorem removes the response of the coefficients $\mC$, but the derivative of the basis functions themselves remains.
This is the Pulay force \citep{Pulay1969Forces}.
It requires the nuclear derivative of every one- and two-electron integral, and it does not vanish as the optimization converges, because a fixed basis is not variationally optimal with respect to its own position.
A floating basis removes it \citep{Hurley1954Electrostatic}.
The splat centers are optimized, so the motion of the basis is part of the inner problem and the envelope theorem removes it together with $\mC$.
What remains is the electrostatic force of \citet{Feynman1939Forces}.

Two complications keep the envelope argument from applying verbatim.
First, the optimizer stops at a finite residual, so we need to know how the force error depends on it.
Second, the minimizer of $E(\RR, \cdot)$ is never isolated.
The energy depends on $\mC$ only through the density matrix, and
\begin{equation*}
    \bar{\mC}\bar{\mC}^\top = \mC\,(\mC^\top \mS \mC)^{-1} \mC^\top
\end{equation*}
is unchanged under $\mC \mapsto \mC \mB$ for every invertible $\mB \in \mathbbm{R}^{N_{\mathrm{occ}} \times N_{\mathrm{occ}}}$.
The quaternion sign of \cref{app:proofs-symmetry} and relabelings of the splats are further, discrete symmetries.
The Hessian $\mH = \partial^2_\theta E$ is therefore singular at every minimizer, and the implicit function theorem cannot be applied to $\theta$ directly.
We work modulo these symmetries.

\begin{assumption}[Nondegenerate minimum up to symmetry]\label{ass:nondegenerate}
Near a minimizer of $E(\RR_0, \cdot)$, the energy is three times continuously differentiable, and for every $\RR$ near $\RR_0$ the minimizers of $E(\RR, \cdot)$ form a smooth manifold $\mathcal{M}_\RR$.
At every $\theta^\star \in \mathcal{M}_\RR$, the Hessian $\mH$ is positive definite on the normal space
\begin{equation*}
    N_{\theta^\star} = (T_{\theta^\star}\mathcal{M}_\RR)^\perp ,
\end{equation*}
with smallest eigenvalue at least $\kappa > 0$.
\end{assumption}

The manifold $\mathcal{M}_\RR$ contains the symmetry orbit of any one minimizer, together with any other flat directions, such as the parameters of a splat that carries no weight.
We write $\mH^+$ for the pseudo-inverse of $\mH$, which inverts it on $N_{\theta^\star}$, and $\mJ = \partial_\RR \partial_\theta E$ for the mixed second derivative.

\begin{proposition}[Floating-basis forces and their residual error]\label{prop:forces}
Let \cref{ass:nondegenerate} hold.
\begin{enumerate}
\item[(i)] The surface $E^\star$ is differentiable, and at any $\theta^\star \in \mathcal{M}_\RR$ the force is the explicit partial derivative, with no Pulay term:
\begin{align}
    \mF^\star_a
    &= -\frac{\mathrm{d} E^\star}{\mathrm{d} \RR_a} \nonumber \\
    &= -\frac{\partial E}{\partial \RR_a}\bigg|_{\theta^\star} \nonumber \\
    &= Z_a \int \dd[3]{\rr}\, \rho^\star(\rr)\, \frac{\rr - \RR_a}{|\rr - \RR_a|^3}
    + \sum_{b \neq a} Z_a Z_b\, \frac{\RR_a - \RR_b}{|\RR_a - \RR_b|^3} . \label{eq:hf-force}
\end{align}
\item[(ii)] Let $\tilde\theta$ be close to $\mathcal{M}_\RR$, and define its residual and computed force as
\begin{align*}
    \bm{\varepsilon} &= \partial_\theta E(\RR, \tilde\theta) , \\
    \tilde{\mF}_a &= -\frac{\partial E}{\partial \RR_a}\bigg|_{\tilde\theta} .
\end{align*}
Then
\begin{align}
    \tilde{\mF} - \mF^\star &= (\partial_\RR \theta^\star)^\top \bm{\varepsilon} + \mathcal{O}(\|\bm{\varepsilon}\|^2), \label{eq:force-residual} \\
    E(\RR, \tilde\theta) - E^\star(\RR) &= \tfrac12\, \bm{\varepsilon}^\top \mH^+ \bm{\varepsilon} + \mathcal{O}(\|\bm{\varepsilon}\|^3), \label{eq:energy-residual}
\end{align}
where
\begin{equation*}
    \partial_\RR \theta^\star = -\mH^+ \mJ^\top
\end{equation*}
is the response of the minimizer normal to $\mathcal{M}_\RR$.
In particular,
\begin{equation*}
    \|\tilde{\mF} - \mF^\star\| \leq \kappa^{-1} \|\mJ\|\, \|\bm{\varepsilon}\| + \mathcal{O}(\|\bm{\varepsilon}\|^2).
\end{equation*}
\item[(iii)] For every $\tilde\theta$, not only near a minimizer, the net computed force equals the summed center components of the residual,
\begin{equation}
\label{eq:net-force}
    \sum_a \tilde{\mF}_a = \sum_\mu \frac{\partial E}{\partial \vm_\mu}\bigg|_{\tilde\theta} .
\end{equation}
\end{enumerate}
\end{proposition}

\begin{proof}
(i) Fix $\theta^\star \in \mathcal{M}_\RR$ and restrict $\theta$ to the affine slice $\theta^\star + N_{\theta^\star}$.
On the slice the Hessian is invertible, so the implicit function theorem applied to the normal component of $\partial_\theta E(\RR, \theta) = 0$ gives a continuously differentiable map $\RR \mapsto \theta^\star(\RR) \in \mathcal{M}_\RR$ through $\theta^\star$.
By the chain rule,
\begin{align*}
    \frac{\mathrm{d}E^\star}{\mathrm{d}\RR}
    &= \partial_\RR E + (\partial_\RR \theta^\star)^\top \partial_\theta E \\
    &= \partial_\RR E ,
\end{align*}
since $\partial_\theta E = 0$ on $\mathcal{M}_\RR$.
Every point of $\mathcal{M}_\RR$ lies on such a map and $\mathrm{d}E^\star/\mathrm{d}\RR$ is unique, so $\partial_\RR E$ takes the same value on all of $\mathcal{M}_\RR$.
The splats float free of the nuclei, so the kinetic, Hartree and exchange--correlation terms depend on $\RR$ only through $\theta$, and so does the auxiliary basis of \cref{eq:adaptive-density-fitting}.
The explicit dependence is confined to $E_{\mathrm{BG}}$ and to $E_{\mathrm{ext}}$, through the potential
\begin{equation*}
    v_{\mathrm{ext}}(\rr) = -\sum_b \frac{Z_b}{|\rr - \RR_b|} .
\end{equation*}
Differentiating these two terms gives the last line of \cref{eq:hf-force}.

(ii) Let $\theta^\star$ be the point of $\mathcal{M}_\RR$ closest to $\tilde\theta$, so that $\bm{\delta} = \tilde\theta - \theta^\star \in N_{\theta^\star}$.
Differentiating $\partial_\theta E = 0$ along $\mathcal{M}_\RR$ shows $T_{\theta^\star}\mathcal{M}_\RR \subseteq \ker \mH$.
Since $\mH$ is symmetric, its range lies in $N_{\theta^\star}$, where it is invertible with $\|\mH^+\| \leq \kappa^{-1}$.
Taylor expansion of the gradient gives
\begin{align*}
    \bm{\varepsilon} &= \mH\bm{\delta} + \mathcal{O}(\|\bm{\delta}\|^2) , \\
    \bm{\delta} &= \mH^+\bm{\varepsilon} + \mathcal{O}(\|\bm{\varepsilon}\|^2) .
\end{align*}
Expanding the explicit partial in the same way,
\begin{align*}
    \partial_\RR E(\RR, \tilde\theta) &= \partial_\RR E(\RR, \theta^\star) + \mJ \bm{\delta} + \mathcal{O}(\|\bm{\delta}\|^2) , \\
    \tilde{\mF} - \mF^\star &= -\mJ \mH^+ \bm{\varepsilon} + \mathcal{O}(\|\bm{\varepsilon}\|^2) .
\end{align*}
Differentiating $\partial_\theta E(\RR, \theta^\star(\RR)) = 0$ in $\RR$ gives
\begin{equation*}
    \mJ^\top + \mH\, \partial_\RR\theta^\star = 0 ,
\end{equation*}
whose solution in $N_{\theta^\star}$ is $\partial_\RR\theta^\star = -\mH^+\mJ^\top$.
Its transpose is $-\mJ\mH^+$, which proves \cref{eq:force-residual}.
For the energy,
\begin{align*}
    E(\RR, \tilde\theta) - E(\RR, \theta^\star)
    &= \tfrac12\, \bm{\delta}^\top \mH \bm{\delta} + \mathcal{O}(\|\bm{\delta}\|^3) \\
    &= \tfrac12\, \bm{\varepsilon}^\top \mH^+ \bm{\varepsilon} + \mathcal{O}(\|\bm{\varepsilon}\|^3) ,
\end{align*}
which is \cref{eq:energy-residual}.

(iii) By \cref{prop:equivariance} with $\mY = \mathbbm{1}$, translating every nucleus and every splat center by the same $\vb$ leaves $E$ unchanged, for any $\tilde\theta$:
\begin{equation*}
    E\big(\{\RR_a + \vb\},\, \{\vm_\mu + \vb\},\, \ldots\big) = E(\RR, \tilde\theta) .
\end{equation*}
Differentiating at $\vb = 0$ gives
\begin{equation*}
    \sum_a \frac{\partial E}{\partial \RR_a} + \sum_\mu \frac{\partial E}{\partial \vm_\mu} = 0 ,
\end{equation*}
which is \cref{eq:net-force}.
\end{proof}

The proposition treats the auxiliary basis as a function of the cloud, which it is at every refresh.
Between refreshes, the optimizer follows the gradient with the auxiliary basis held fixed (\cref{alg:gsdft}).
That gradient omits the dependence of the fit on the auxiliary functions, and the omitted term is proportional to the Coulomb-norm fit residual that \cref{app:proofs-df} controls.

\Cref{prop:forces} allows three statements in the main text.
First, a force costs one reverse-mode pass through two terms of the energy, with no derivative integrals and no Pulay correction.
Second, the energy error is quadratic in the residual and the force error is linear.
Along one optimization, forces therefore converge more slowly than energies even at a fixed cloud size, so force labels need a tighter stopping criterion than energies do, and the gradient norm that the optimizer already reports bounds their error.
Third, \cref{eq:net-force} is a free and reference-free check.
An exact force field has zero net force, and the computed net force is exactly the summed center gradient, so a nonzero net force measures incomplete optimization and nothing else.

\subsection{Adaptive density fitting}\label{app:proofs-df}

The Hartree energy is a six-dimensional integral that is quartic in the basis (\cref{eq:eri}), and density fitting replaces it with a quadratic one by expanding the density in an auxiliary basis.
In GS-DFT, the fitted energy sits inside the objective that the optimizer minimizes.
This raises two questions.
The first is how accurate the fit is.
The second is familiar from learning with a surrogate inside the loss: can the optimizer lower the objective by degrading the surrogate instead of improving the physics?
\Cref{prop:dunlap} answers the first for any auxiliary basis, and \cref{prop:pairproduct} answers the second for the adaptive basis of \cref{sec:method}.

The object that makes both answers exact is the Coulomb inner product,
\begin{align}
    (f \,|\, g)
    &= \int \dd[3]{\rr} \int \dd[3]{\rr'}\, \frac{f(\rr)\, g(\rr')}{|\rr - \rr'|} \nonumber \\
    &= \int \frac{\dd[3]{\vk}}{(2\pi)^3}\, \frac{4\pi}{|\vk|^2}\, \overline{\hat f(\vk)}\, \hat g(\vk), \label{eq:coulomb-inner}
\end{align}
defined for real functions of finite Coulomb energy, with $\hat f$ the Fourier transform.
The kernel $4\pi / |\vk|^2$ is strictly positive, so $(\cdot|\cdot)$ is an inner product with norm $\|f\|_{\mathrm{C}} = (f|f)^{1/2}$.
In this notation, the Hartree energy of \cref{eq:energy-terms} and the quantities of \cref{eq:adaptive-density-fitting} read
\begin{align*}
    E_H[\rho] &= \tfrac12\, \|\rho\|_{\mathrm{C}}^2 , \\
    \mQ_{ab} &= (h_a \,|\, h_b) , \\
    \vt_a &= (\rho \,|\, h_a) .
\end{align*}
$\mQ$ is the Gram matrix of the auxiliary functions, and it is positive definite whenever they are linearly independent.
The density fit of \cref{sec:background} is therefore least squares in the Coulomb norm, and its solution defines the Coulomb-orthogonal projection of $\rho$ onto $\operatorname{span}\{h_a\}$,
\begin{align*}
    \vd &= \mQ^{-1}\vt , \\
    \Pi \rho &= \textstyle\sum_a d_a h_a .
\end{align*}

\begin{proposition}[Robust density fitting \citep{Dunlap2000Robust}]\label{prop:dunlap}
Let $\{h_a\}_{a=1}^{N_{\mathrm{aux}}}$ be linearly independent functions of finite Coulomb energy, and let $\rho_{\vd} = \sum_a d_a h_a$ for $\vd \in \mathbbm{R}^{N_{\mathrm{aux}}}$.
For every $\vd$,
\begin{equation}
\label{eq:robust}
    \vt^\top \vd - \tfrac12\, \vd^\top \mQ\, \vd = E_H[\rho] - \tfrac12\, \|\rho - \rho_{\vd}\|_{\mathrm{C}}^2 .
\end{equation}
The left-hand side is therefore a lower bound on $E_H[\rho]$, with an error that is quadratic in the fit error.
It is maximized at $\vd = \mQ^{-1}\vt$, where it equals
\begin{equation*}
    \tilde{E}_H = \tfrac12\, \vt^\top \mQ^{-1} \vt ,
\end{equation*}
and the gap to the exact Hartree energy is
\begin{equation}
\label{eq:df-gap}
    E_H[\rho] - \tilde{E}_H = \tfrac12\, \|\rho - \Pi\rho\|_{\mathrm{C}}^2 \;\geq\; 0 .
\end{equation}
For $\lambda \geq 0$, the regularized estimate satisfies
\begin{align*}
    \tilde{E}_H^\lambda &= \tfrac12\, \vt^\top (\mQ + \lambda\mathbbm{1})^{-1} \vt , \\
    0 \;\leq\; \tilde{E}_H - \tilde{E}_H^\lambda &\leq \frac{\lambda}{\lambda_{\min}(\mQ) + \lambda}\, \tilde{E}_H .
\end{align*}
\end{proposition}

\begin{proof}
Expanding the square,
\begin{align*}
    \tfrac12\|\rho - \rho_{\vd}\|_{\mathrm{C}}^2
    &= \tfrac12 (\rho|\rho) - \sum_a d_a (\rho|h_a) + \tfrac12 \sum_{ab} d_a d_b (h_a|h_b) \\
    &= E_H[\rho] - \vt^\top\vd + \tfrac12\, \vd^\top\mQ\vd ,
\end{align*}
which rearranges to \cref{eq:robust}.
The left-hand side is strictly concave in $\vd$, with stationary point $\mQ\vd = \vt$.
Substituting $\vd = \mQ^{-1}\vt$ gives $\tilde{E}_H$, and the right-hand side of \cref{eq:robust} at this point gives \cref{eq:df-gap}.
For the regularized estimate,
\begin{equation*}
    \mQ^{-1} - (\mQ + \lambda\mathbbm{1})^{-1} = \lambda\, \mQ^{-1}(\mQ + \lambda\mathbbm{1})^{-1} .
\end{equation*}
In the eigenbasis of $\mQ$, with eigenvalues $q_i$, the right-hand side is diagonal with entries $q_i^{-1}\, \lambda / (q_i + \lambda)$, so
\begin{equation*}
    0 \;\preceq\; \mQ^{-1} - (\mQ + \lambda\mathbbm{1})^{-1} \;\preceq\; \frac{\lambda}{\lambda_{\min}(\mQ) + \lambda}\, \mQ^{-1} .
\end{equation*}
Taking the quadratic form in $\vt$ gives the bound.
\end{proof}

\Cref{eq:robust} is why the fit is called robust.
A first-order error in the fitted coefficients produces only a second-order error in the energy, because the energy is stationary in $\vd$ at the fit.
\Cref{eq:df-gap} makes the error a geometric quantity: half the squared Coulomb distance from $\rho$ to the auxiliary span.
The regularizer $\lambda = 10^{-8}$ used in all experiments keeps the Cholesky factorization of \cref{alg:gsdft} stable when two auxiliary functions nearly coincide, and it can only lower the estimate further.

For the adaptive basis, the span is built from the density itself.
By the Gaussian product theorem, the product of two splats is a single Gaussian,
\begin{equation}
\label{eq:gpt}
    g_\mu(\rr)\, g_\nu(\rr) \propto h_{\mu\nu}(\rr) ,
\end{equation}
with the proportionality constant of \cref{eq:ovlp}, where $h_{\mu\nu}$ has precision and center
\begin{align*}
    \mA_{\mu\nu} &= \mA_\mu + \mA_\nu , \\
    \vm_{\mu\nu} &= (\mA_\mu + \mA_\nu)^{-1}(\mA_\mu \vm_\mu + \mA_\nu \vm_\nu) .
\end{align*}
The density $\rho = \sum_{\mu\nu} \mP_{\mu\nu}\, g_\mu g_\nu$ of \cref{eq:density-matrix} is a linear combination of these products.
Screening retains a set of unordered pairs and splits the density into a retained and a discarded part,
\begin{align*}
    \mathcal{P} &= \big\{ \{\mu,\nu\} : |\mS_{\mu\nu}| > \tau \big\} , \\
    \rho &= \rho_{\mathcal{P}} + \rho_{\mathcal{P}^c} .
\end{align*}
The auxiliary basis is refreshed every $K$ steps and held fixed in between (\cref{alg:gsdft}).
We write $\Pi_r$ for the projection onto the span built from the cloud $\Theta_r$ at the most recent refresh, and $\rho_r$ for the density at that refresh.

\begin{proposition}[Closure of the adaptive fit]\label{prop:pairproduct}
At a refresh, the fit error is confined to the discarded pairs,
\begin{align}
    E_H[\rho_r] - \tilde{E}_H
    &= \tfrac12\, \|(1 - \Pi_r)\, \rho_{r, \mathcal{P}^c}\|_{\mathrm{C}}^2 \nonumber \\
    &\leq \tfrac12\, \|\rho_{r, \mathcal{P}^c}\|_{\mathrm{C}}^2 , \label{eq:closure}
\end{align}
and the fit is exact for $\tau = 0$.
At any later step before the next refresh, with density $\rho$,
\begin{equation}
\label{eq:drift}
    \Big( E_H[\rho] - \tilde{E}_H \Big)^{1/2} \;\leq\; \tfrac{1}{\sqrt 2} \Big( \|\rho_{r, \mathcal{P}^c}\|_{\mathrm{C}} + \|\rho - \rho_r\|_{\mathrm{C}} \Big) .
\end{equation}
The same statements hold for each exchange pair density $\rho_{kl} = \bar\psi_k \bar\psi_l$.
With $(\vt_{kl})_a = (\rho_{kl} | h_a)$, the fitted and exact exchange energies are
\begin{align*}
    \tilde{E}_x &= -\sum_{kl} \vt_{kl}^\top \mQ^{-1} \vt_{kl} , \\
    E_x &= -\sum_{kl} \|\rho_{kl}\|_{\mathrm{C}}^2 ,
\end{align*}
and the fit is an upper bound, $\tilde{E}_x \geq E_x$.
\end{proposition}

\begin{proof}
Every retained product is proportional to an auxiliary function, so $\rho_{r, \mathcal{P}} \in \operatorname{span}\{h_a\}$ and
\begin{align*}
    \Pi_r \rho_{r, \mathcal{P}} &= \rho_{r, \mathcal{P}} , \\
    \rho_r - \Pi_r \rho_r &= (1 - \Pi_r)\, \rho_{r, \mathcal{P}^c} .
\end{align*}
\Cref{eq:df-gap} then gives the identity in \cref{eq:closure}.
The inequality holds because an orthogonal projection does not increase the norm, and for $\tau = 0$ the discarded part is empty.
For \cref{eq:drift}, write $\rho = \rho_r + (\rho - \rho_r)$ and use the triangle inequality together with $\|(1 - \Pi_r) f\|_{\mathrm{C}} \leq \|f\|_{\mathrm{C}}$,
\begin{align*}
    \|(1 - \Pi_r)\, \rho\|_{\mathrm{C}}
    &\leq \|(1 - \Pi_r)\, \rho_r\|_{\mathrm{C}} + \|(1 - \Pi_r)(\rho - \rho_r)\|_{\mathrm{C}} \\
    &\leq \|\rho_{r, \mathcal{P}^c}\|_{\mathrm{C}} + \|\rho - \rho_r\|_{\mathrm{C}} ,
\end{align*}
and apply \cref{eq:df-gap}.
For exchange, the L\"owdin orbitals are linear combinations of splats, so each $\rho_{kl}$ is again a linear combination of splat products and the same argument applies.
Applying \cref{eq:df-gap} to each $\rho_{kl}$ gives
\begin{equation*}
    \|\rho_{kl}\|_{\mathrm{C}}^2 \geq \vt_{kl}^\top\mQ^{-1}\vt_{kl} ,
\end{equation*}
and the minus sign in $E_x$ reverses the inequality.
\end{proof}

The answer to the second question follows.
The objective the optimizer minimizes is the exact-integral energy minus the gap of \cref{eq:df-gap}, so a step that enlarges the gap lowers the objective without improving the physics.
The fit coefficients are solved in closed form, so there is no inner adversary to train.
The only exploitable quantity is the gap itself, and \cref{prop:pairproduct} bounds it.
At every refresh it resets to the contribution of the discarded pairs, which carry charge $|\mP_{\mu\nu} \mS_{\mu\nu}| \leq \tau |\mP_{\mu\nu}|$ each.
Between refreshes it can grow only with the distance the density has moved since the last refresh.
This is the trade-off that the refresh period $K$ controls: a larger $K$ amortizes the cubic factorization of \cref{alg:gsdft} over more steps, and allows the density to drift further from the span it was fitted in.

Two consequences for the reported energies follow.
For semilocal functionals only the Hartree term is fitted, so the fitted energy lies below the exact-integral energy at the same parameters, by the gap.
For hybrid functionals and Hartree--Fock, the Hartree and exchange errors have opposite signs, and the fitted energy is not a bound in either direction.
In both cases the error vanishes as $\tau \to 0$ and $K \to 1$.

\subsection{Regularized eigensolve}\label{app:proofs-regeigh}

The L\"owdin retraction of \cref{eq:lowdin} needs the inverse square root of the occupied Gram matrix
\begin{equation*}
    \mG = \mC^\top \mS \mC ,
\end{equation*}
and both its value and its gradient can fail.
In the forward pass, $\mG^{-1/2}$ diverges as the smallest eigenvalue of $\mG$ goes to zero.
In a floating basis this is not an edge case.
The energy is invariant under $\mC \mapsto \mC\mB$ (\cref{app:proofs-forces}), so nothing in the objective keeps the columns of $\mC$ from becoming nearly parallel, and splats that drift on top of each other make $\mS$ itself nearly singular.
In the backward pass, the standard gradient of an eigendecomposition contains the factors $(\lambda_i - \lambda_j)^{-1}$, which diverge at degenerate eigenvalues even when the forward pass is harmless.
This is not an edge case either: an orthonormal initialization gives $\mG = \mathbbm{1}$, where every pair of eigenvalues is degenerate.
The same failure appears in machine learning wherever a network differentiates through an eigendecomposition, as in whitening layers.
\Cref{prop:regeigh} states what the regularization of \cref{eq:regeigh} changes, and shows that the divergence of the backward pass is an artifact of how the gradient is assembled rather than a property of the derivative.

What we need is the derivative of a matrix function.
For a symmetric $\mG = \mU \diag(\bm{\lambda})\, \mU^\top$ and a scalar function $f$ differentiable on the spectrum, the matrix function $f(\mG) = \mU \diag(f(\bm{\lambda}))\, \mU^\top$ has the Fr\'echet derivative \citep[Thm.~V.3.3]{Bhatia1997MatrixAnalysis}
\begin{equation}
\label{eq:daleckii-krein}
    \mathrm{D}f(\mG)[\dot{\mG}] = \mU \big( \bm{\Omega} \circ \mU^\top \dot{\mG}\, \mU \big)\, \mU^\top ,
\end{equation}
where $\circ$ is the elementwise product and the kernel is the first divided difference $f^{[1]}$ of $f$,
\begin{align*}
    \Omega_{ij} &= f^{[1]}(\lambda_i, \lambda_j) , \\
    f^{[1]}(\lambda_i, \lambda_j) &=
    \begin{cases}
        \dfrac{f(\lambda_i) - f(\lambda_j)}{\lambda_i - \lambda_j}, & \lambda_i \neq \lambda_j, \\[1.5ex]
        f'(\lambda_i), & \lambda_i = \lambda_j .
    \end{cases}
\end{align*}
By the mean value theorem, for some $\xi$ between $\lambda_i$ and $\lambda_j$,
\begin{equation*}
    f^{[1]}(\lambda_i, \lambda_j) = f'(\xi) ,
\end{equation*}
so $\bm{\Omega}$ is bounded wherever $f'$ is.
The derivative itself therefore has no singularity at degenerate eigenvalues.
The standard eigendecomposition gradient differentiates $\mU$ and $\bm{\lambda}$ separately, and each of those derivatives is singular at a degeneracy, even though their combination in \cref{eq:daleckii-krein} is not.

The forward pass of \cref{eq:regeigh} uses
\begin{align*}
    f(\lambda) &= \max(\lambda, \varphi)^{-1/2} , \\
    \varphi &= \varepsilon \lambda_{\max} , \\
    \mT &= f(\mG) , \\
    \bar{\mC} &= \mC\mT ,
\end{align*}
and we call an eigenvalue floored when $\lambda_i < \varphi$.
The regularized backward pass replaces $\bm{\Omega}$ by
\begin{equation*}
    \tilde\Omega_{ij} =
    \begin{cases}
        0, & \lambda_i < \varphi \text{ or } \lambda_j < \varphi, \\[0.5ex]
        \tfrac12 \big( f'(\lambda_i) + f'(\lambda_j) \big), & |\lambda_i - \lambda_j| \leq \delta \lambda_{\max}, \\[0.5ex]
        \Omega_{ij}, & \text{otherwise}.
    \end{cases}
\end{equation*}
All derivatives treat the floor level $\varphi$ as a constant.

\begin{proposition}[Regularized L\"owdin retraction]\label{prop:regeigh}
Let $\mS$ be positive semidefinite, so that $\mG \succeq 0$, and let $\lambda_{\max} > 0$.
\begin{enumerate}
\item[(i)] If $\lambda_{\min}(\mG) \geq \varepsilon\lambda_{\max}$, that is, if $\operatorname{cond}(\mG) \leq \varepsilon^{-1}$, then exactly
\begin{align*}
    \mT &= \mG^{-1/2} , \\
    \bar{\mC}^\top \mS \bar{\mC} &= \mathbbm{1} .
\end{align*}
\item[(ii)] In general, $\|\mT\| \leq \varphi^{-1/2}$, and the electron count of $\mP = 2\bar{\mC}\bar{\mC}^\top$ is
\begin{align}
    \operatorname{Tr}(\mP\mS)
    &= 2 \sum_{i=1}^{N_{\mathrm{occ}}} \min\!\Big(1, \frac{\lambda_i}{\varphi}\Big) \nonumber \\
    &\leq 2 N_{\mathrm{occ}} , \label{eq:charge-deficit}
\end{align}
with equality if and only if no eigenvalue is floored.
\item[(iii)] For every pair of unfloored eigenvalues, $\tilde\Omega_{ij} = \Omega_{ij}$ if $|\lambda_i - \lambda_j| > \delta\lambda_{\max}$, and otherwise
\begin{equation}
\label{eq:midpoint-error}
    \big| \tilde\Omega_{ij} - \Omega_{ij} \big| \;\leq\; \frac{5}{32}\, \frac{(\lambda_i - \lambda_j)^2}{\min(\lambda_i, \lambda_j)^{7/2}} .
\end{equation}
To leading order, the relative error is
\begin{equation*}
    \frac{\big| \tilde\Omega_{ij} - \Omega_{ij} \big|}{|\Omega_{ij}|} \approx \frac{5}{16}\, \frac{(\lambda_i - \lambda_j)^2}{\lambda_i^2} .
\end{equation*}
\item[(iv)] The regularized gradient is the exact gradient with respect to perturbations of $\mG$ restricted to the unfloored eigenspace, up to the error in (iii).
If no eigenvalue is floored and no two eigenvalues lie within $\delta\lambda_{\max}$ of each other, it is the exact gradient.
\end{enumerate}
\end{proposition}

\begin{proof}
(i) If no eigenvalue is floored, $f(\lambda) = \lambda^{-1/2}$ on the spectrum, so $\mT = \mG^{-1/2}$ and
\begin{align*}
    \bar{\mC}^\top\mS\bar{\mC}
    &= \mT\, \mC^\top \mS \mC\, \mT \\
    &= \mG^{-1/2}\, \mG\, \mG^{-1/2} \\
    &= \mathbbm{1} .
\end{align*}

(ii) The eigenvalues of $\mT$ are $\max(\lambda_i, \varphi)^{-1/2} \leq \varphi^{-1/2}$.
Since $\mT$ and $\mG$ share eigenvectors,
\begin{align*}
    \bar{\mC}^\top\mS\bar{\mC}
    &= \mT\mG\mT \\
    &= \mU \diag\!\Big(\frac{\lambda_i}{\max(\lambda_i, \varphi)}\Big)\, \mU^\top ,
\end{align*}
and $\operatorname{Tr}(\mP\mS) = 2\operatorname{Tr}(\bar{\mC}^\top\mS\bar{\mC})$ gives \cref{eq:charge-deficit}.
Each term is at most one, with equality exactly when $\lambda_i \geq \varphi$.

(iii) Above the floor, $f(\lambda) = \lambda^{-1/2}$.
The divided difference is the average of $f'$ over the interval between $\lambda_j$ and $\lambda_i$, and $\tilde\Omega_{ij}$ is the trapezoidal estimate of that average.
On an interval of length $h$, the trapezoidal rule integrates $f'$ with error $h^3 |f'''(\xi)| / 12$ for some $\xi$ in the interval, so the error in the average is $h^2 |f'''(\xi)| / 12$.
The third derivative and the kernel are
\begin{align*}
    |f'''(\lambda)| &= \tfrac{15}{8}\, \lambda^{-7/2} , \\
    |\Omega_{ij}| &\approx |f'(\lambda_i)| = \tfrac12\, \lambda_i^{-3/2} .
\end{align*}
The first is largest at the smaller eigenvalue, which gives \cref{eq:midpoint-error}, and dividing by the second gives the relative error.
On the diagonal, $\tilde\Omega_{ii} = f'(\lambda_i) = \Omega_{ii}$ exactly.

(iv) Let $\mU_+$ collect the eigenvectors of the unfloored eigenvalues.
Because $\tilde\Omega_{ij} = 0$ whenever $i$ or $j$ is floored, the regularized backward pass maps an upstream gradient $\bar{\mT}$ to
\begin{equation*}
    \bar{\mT} \;\longmapsto\; \mU_+ \big( \tilde{\bm{\Omega}}_{++} \circ \mU_+^\top \bar{\mT}\, \mU_+ \big)\, \mU_+^\top .
\end{equation*}
By \cref{eq:daleckii-krein}, this is the adjoint of the restricted derivative
\begin{equation*}
    \dot{\mG} \;\longmapsto\; \mathrm{D}f(\mG)\big[\mU_+\mU_+^\top \dot{\mG}\, \mU_+\mU_+^\top\big] ,
\end{equation*}
with $\bm{\Omega}_{++}$ replaced by $\tilde{\bm{\Omega}}_{++}$.
By (iii), the two agree when no two unfloored eigenvalues lie within $\delta\lambda_{\max}$, and if no eigenvalue is floored the projection is the identity.
\end{proof}

The proposition separates what the regularization does in the two passes.
The forward floor is inactive on a well-conditioned Gram matrix, and when it engages it can only remove charge, never add it, by the amount in \cref{eq:charge-deficit}.
The backward regularization is exact on the diagonal and between floored eigenvalues, and within a near-degenerate pair it makes a second-order error that is small unless the pair sits close to the floor.
The one deliberate change to the gradient is that the directions coupling floored and unfloored modes are dropped.
This is a choice we made to stop the gradient from acting on the orbital combinations that the retraction is failing to normalize.

Two consequences follow.
First, at convergence the regularization has no effect whenever the converged Gram matrix has condition number below $\varepsilon^{-1}$ and no near-degenerate eigenvalues.
The optimizer then solves the exact stationarity condition, and \cref{prop:forces} applies unchanged.
Second, whether this holds is checkable after the fact from the logged ratio $\lambda_{\min}(\mG) / \lambda_{\max}(\mG)$, and any floor that remains active shows up as the electron deficit of \cref{eq:charge-deficit}.

\section{From Gaussian-type orbitals to floating splats}\label{appendix:representation}

Quantum chemistry has refined fixed Gaussian basis sets for seventy years, and computer graphics has recently shown that a cloud of free anisotropic Gaussians, optimized end to end, can represent a complicated three-dimensional signal.
This appendix makes the connection precise in both directions.
\Cref{app:gto} defines a contracted Gaussian basis and counts its parameters, which is the accounting behind every comparison in \cref{sec:experiments}.
\Cref{app:gto-limit} shows that splats contain Gaussian-type orbitals as limits, and what that inclusion costs.
\Cref{app:fsgo} places the earlier floating-orbital models inside the GS-DFT family and identifies what stopped them.
\Cref{app:3dgs} states the correspondence with 3DGS and proves where it stops.

\subsection{Contracted Gaussian bases}\label{app:gto}

Every Gaussian baseline in \cref{sec:experiments} belongs to a standard basis-set family, and every comparison there is made either at matched function count or per free parameter.
Both comparisons require a precise definition of a basis function and of the numbers that define it.

The most basic building block is a primitive Gaussian pinned to a nucleus $\RR_a$.
Writing $\vs = \rr - \RR_a$, a Cartesian primitive of angular momentum $l = n_x + n_y + n_z$ and exponent $\alpha > 0$ is
\begin{equation}
\label{eq:primitive}
    \chi^{\mathrm{prim}}(\rr) = s_x^{n_x}\, s_y^{n_y}\, s_z^{n_z}\; e^{-\alpha |\vs|^2} .
\end{equation}
The $(l+1)(l+2)/2$ monomials of degree $l$ are usually replaced by the $2l+1$ real solid harmonics
\begin{equation*}
    \mathcal{Y}_{lm}(\vs) = |\vs|^l\, Y_{lm}(\vs / |\vs|) , \qquad m = -l, \ldots, l ,
\end{equation*}
which span the part of angular momentum $l$.

\begin{definition}[Contracted Gaussian basis]\label{def:gto-basis}
A shell $s$ on atom $a(s)$ consists of an angular momentum $l_s$, $K_s$ exponents $\alpha_{sp} > 0$ and $K_s$ contraction weights $d_{sp}$.
It contributes the $2l_s + 1$ functions
\begin{equation}
\label{eq:lcao}
    \chi_{sm}(\rr) = \mathcal{Y}_{l_s m}(\rr - \RR_{a(s)}) \sum_{p=1}^{K_s} d_{sp}\; e^{-\alpha_{sp} |\rr - \RR_{a(s)}|^2} ,
    \qquad m = -l_s, \ldots, l_s .
\end{equation}
A basis set is a table that assigns a list of shells to every element.
For a molecule it yields $B = \sum_s (2l_s + 1)$ functions, and the orbitals are expanded as
\begin{equation*}
    \psi_k(\rr) = \sum_{\mu=1}^{B} \mC_{\mu k}\, \chi_\mu(\rr) .
\end{equation*}
\end{definition}

All $2l_s + 1$ functions of a shell share the same exponents and weights, so a shell is defined by $2K_s$ numbers, whatever its angular momentum.
The exponents and weights are fitted once per element, largely to atomic calculations, and then stored \citep{Boys1950Gaussian, ditchfield1971self, dunning1989gaussian}.
During a calculation, only $\mC$ is optimized.

Basis sets come in families indexed by a cardinal number $\zeta$ \citep{jensenAtomicOrbitalBasis2013}.
In the correlation-consistent family cc-pV$\zeta$Z \citep{dunning1989gaussian}, with $\zeta = 2, 3, 4$ for D, T, Q, a first-row atom carries
\begin{equation*}
    n_l(\zeta) = \zeta + 1 - l , \qquad l = 0, \ldots, \zeta ,
\end{equation*}
contracted shells of angular momentum $l$, which gives $3s\,2p\,1d$, $4s\,3p\,2d\,1f$ and $5s\,4p\,3d\,2f\,1g$.
Each step in $\zeta$ adds one shell of every angular momentum already present and raises the highest angular momentum by one.
The number of functions per first-row atom is therefore
\begin{align*}
    B_{\mathrm{atom}}(\zeta)
    &= \sum_{l=0}^{\zeta} (\zeta + 1 - l)(2l + 1) \\
    &= \tfrac16 (\zeta + 1)(\zeta + 2)(2\zeta + 3) ,
\end{align*}
that is $14$, $30$ and $55$, and hydrogen carries $B_{\mathrm{atom}}(\zeta - 1)$.
The basis grows cubically in $\zeta$.
For self-consistent-field energies such as Hartree--Fock and DFT, the energy converges close to exponentially in $\zeta$ \citep{jensenAtomicOrbitalBasis2013},
\begin{equation*}
    E_\zeta \approx E_{\mathrm{CBS}} + A\, e^{-b\zeta} ,
\end{equation*}
and this regularity is what makes an extrapolation to the complete-basis-set (CBS) limit $E_{\mathrm{CBS}}$ meaningful.
The def2 family \citep{Weigend2005Def2} uses segmented contractions designed for balanced accuracy across the periodic table, and we use its quadruple-zeta member def2-QZVPP where a large reference basis is needed.

A tabulated basis is often described as having no parameters, because nothing in it is optimized during a calculation.
We count its tabulated numbers anyway.
They are parameters of the representation in the machine-learning sense: values that had to be chosen, fitted to data, and shipped with the model for its outputs to mean anything, like the weights of a pretrained network that is used frozen.
With the nine-parameter chart of \cref{eq:splat}, the parameter counts of a basis of $B$ functions and of a cloud of $M$ splats are
\begin{align}
    P_{\mathrm{GTO}} &= 2 \sum_{s} K_s + B\, N_{\mathrm{occ}} , \label{eq:param-gto} \\
    P_{\mathrm{splat}} &= 9 M + M\, N_{\mathrm{occ}} . \label{eq:param-splat}
\end{align}
We write
\begin{equation*}
    \eta = \frac{2 \sum_s K_s}{B}
\end{equation*}
for the tabulated numbers per Gaussian function.
Across the cc-pV$\zeta$Z bases used in this paper, $\eta \in [6.6, 7.5]$, against $9$ for a splat, so both representations pay the same order per function.
The basis part of each count scales with the number of functions and the coefficient part with the number of functions times $N_{\mathrm{occ}}$, so
\begin{equation*}
    \frac{\text{basis parameters}}{\text{coefficient parameters}} = \frac{\eta}{N_{\mathrm{occ}}} ,
\end{equation*}
with $\eta = 9$ for splats.
\Cref{tab:param-share} evaluates this ratio for the three molecules of \cref{sec:experiments}.

\begin{table}[h]
\centering
\caption{Basis parameters per coefficient parameter, $\eta / N_{\mathrm{occ}}$, with $\eta = 7$ for the Gaussian basis and $\eta = 9$ for splats.}
\label{tab:param-share}
\begin{tabular}{lccc}
\toprule
 & \ce{H2O} & \ce{C2H5OH} & Ala dipeptide \\
\midrule
$N_{\mathrm{occ}}$ & 5 & 13 & 39 \\
Gaussian basis & 1.40 & 0.54 & 0.18 \\
Splat cloud & 1.80 & 0.69 & 0.23 \\
\bottomrule
\end{tabular}
\end{table}

At matched function count $M = B$, \cref{eq:param-gto,eq:param-splat} give the relative overhead of the cloud,
\begin{equation}
\label{eq:param-overhead}
    \frac{P_{\mathrm{splat}}}{P_{\mathrm{GTO}}} - 1 = \frac{9 - \eta}{\eta + N_{\mathrm{occ}}} ,
\end{equation}
which decreases monotonically in $N_{\mathrm{occ}}$.
For $\eta \geq 6.6$ it is at most $21\%$ on water and at most $5.3\%$ on alanine dipeptide.
Water with cc-pVDZ is a concrete case, with $B = 24$ functions defined by $158$ tabulated numbers:
\begin{align*}
    P_{\mathrm{GTO}} &= 158 + 24 \cdot 5 \\
                      &= 278 , \\
    P_{\mathrm{splat}} &= 9 \cdot 24 + 24 \cdot 5 \\
                        &= 336 .
\end{align*}
The ratio is $1.21$, while counting only the coefficients on the Gaussian side would give $2.8$.
Matched function count is therefore also matched parameter count, up to the correction of \cref{eq:param-overhead}.
This is why the accuracy-per-parameter comparison of \cref{fig:isoparams} counts the tabulated basis on the Gaussian side, and why a gain at matched function count is not bought with hidden parameters.

\subsection{Gaussian-type orbitals as limits of splats}\label{app:gto-limit}

Every splat in \cref{eq:splat} is a positive function with no angular factor, while a Gaussian basis carries $p$, $d$ and $f$ functions whose angular factors create the nodal structure of an orbital.
\Cref{sec:method} claims that signed coefficients recover this structure from displaced splats of opposite sign.
\Cref{prop:gto-limit} makes the claim precise, and \cref{prop:gto-cost} shows what it costs.

We use isotropic Gaussians, which are splats with precision $\mA = 2\alpha\mathbbm{1}$ up to normalization,
\begin{equation*}
    e_{\vm}(\rr) = e^{-\alpha|\rr - \vm|^2} .
\end{equation*}
Two facts are needed.
First, derivatives with respect to the center produce Hermite polynomials.
For a multi-index $\vn = (n_x, n_y, n_z)$ with $|\vn| = n_x + n_y + n_z$,
\begin{equation}
\label{eq:hermite}
    \partial_{\vm}^{\vn}\, e_{\vm}(\rr) = \alpha^{|\vn|/2} \prod_{i \in \{x,y,z\}} H_{n_i}\!\big(\sqrt{\alpha}\,(r_i - m_i)\big)\; e_{\vm}(\rr) ,
\end{equation}
where $H_n$ is the Hermite polynomial of degree $n$.
This is the Hermite-Gaussian form used to evaluate molecular integrals \citep{Helgaker2000MEST}.
Second, a derivative in the center is a limit of differences of displaced copies.
For a step $h > 0$, the forward difference is
\begin{equation}
\label{eq:forward-diff}
    \Delta_h^{\vn} e_{\vm} = \sum_{\vj \leq \vn} (-1)^{|\vn| - |\vj|} \binom{\vn}{\vj}\, e_{\vm + h\vj} ,
    \qquad
    \binom{\vn}{\vj} = \prod_i \binom{n_i}{j_i} ,
\end{equation}
where $\vj \leq \vn$ holds componentwise.
It combines $(n_x + 1)(n_y + 1)(n_z + 1)$ displaced copies of $e_{\vm}$, all of which lie on the lattice
\begin{equation*}
    T_l = \big\{ \vj \in \mathbbm{N}^3 : |\vj| \leq l \big\} ,
    \qquad
    |T_l| = \tfrac16 (l+1)(l+2)(l+3) ,
\end{equation*}
whenever $|\vn| \leq l$.

\begin{proposition}[Gaussian-type orbitals are limits of splats]\label{prop:gto-limit}
Let $H^1$ be the Sobolev space of square-integrable functions with square-integrable gradient, the natural space for the kinetic energy of \cref{eq:energy-terms}.
\begin{enumerate}
\item[(i)] For every $l \geq 0$, with $\vs = \rr - \vm$,
\begin{equation*}
    \operatorname{span}\big\{ \vs^{\vn} e^{-\alpha|\vs|^2} : |\vn| \leq l \big\}
    = \operatorname{span}\big\{ \partial_{\vm}^{\vn} e_{\vm} : |\vn| \leq l \big\} .
\end{equation*}
\item[(ii)] For every multi-index $\vn$, as $h \to 0$,
\begin{equation*}
    \big\| h^{-|\vn|} \Delta_h^{\vn} e_{\vm} - \partial_{\vm}^{\vn} e_{\vm} \big\|_{H^1} = \mathcal{O}(h) .
\end{equation*}
\item[(iii)] Every function of a shell with angular momentum $l$ and $K$ primitives, for any contraction weights, is an $H^1$ limit of combinations of the $K\,|T_l|$ splats with the same exponents centered at $\RR_a + h\vj$, $\vj \in T_l$.
The same splats serve all $2l + 1$ functions of the shell.
\end{enumerate}
\end{proposition}

\begin{proof}
(i) In one dimension, the Rodrigues formula with $t = \sqrt{\alpha}(x - m)$ reads
\begin{equation*}
    H_n(t)\, e^{-t^2} = (-1)^n\, \frac{\mathrm{d}^n}{\mathrm{d}t^n}\, e^{-t^2} ,
\end{equation*}
and $\partial_m = -\sqrt{\alpha}\, \partial_t$ gives \cref{eq:hermite}.
The three-dimensional case follows because $e_{\vm}$ factorizes over coordinates.
$H_n$ has degree exactly $n$, so the Hermite polynomials and the monomials of degree at most $n$ are related by an invertible triangular change of basis.
Taking products over coordinates, the polynomials multiplying $e_{\vm}$ on the two sides span the same space.

(ii) The map $\vm \mapsto e_{\vm}$ from $\mathbbm{R}^3$ to $H^1$ is infinitely differentiable, because every center derivative is a polynomial times a Gaussian.
For a smooth map into a Banach space, $h^{-|\vn|}\Delta_h^{\vn}$ equals $\partial^{\vn}_{\vm}$ plus a Taylor remainder bounded by $h$ times the supremum of the derivatives of order $|\vn| + 1$ over the box spanned by the lattice points, which is finite.

(iii) A function of the shell is a solid harmonic of degree $l$ times a contracted radial part, and a solid harmonic is a combination of monomials of degree $l$.
By (i), each primitive is a finite combination of center derivatives of order at most $l$.
By (ii), each of these is the $H^1$ limit of combinations of copies centered at $\RR_a + h\vj$, $\vj \in T_l$.
The contraction is a fixed finite sum over the $K$ exponents, and the lattice does not depend on $m$.
\end{proof}

The inclusion extends from functions to energies.
Suppose $E$ is continuous in $H^1$ on orthonormal orbitals, as it is for Hartree--Fock and the local density approximation.
Approximating each orbital within $\delta$ in $H^1$ moves the L\"owdin-orthonormalized orbitals by $\mathcal{O}(\delta)$, because the retraction of \cref{eq:lowdin} is continuous where the Gram matrix is invertible.
Hence, for a basis with shells $(l_s, K_s)$,
\begin{equation}
\label{eq:gto-energy-bound}
    \inf_{\Theta,\, \mC} E \;\leq\; E_{\mathrm{GTO}}
    \qquad \text{over clouds of size} \qquad
    M = \sum_s K_s\, |T_{l_s}| .
\end{equation}

\begin{proposition}[The cost of the limit]\label{prop:gto-cost}
Let $\vc_h$ be the coefficient vector of $h^{-|\vn|}\Delta_h^{\vn} e_{\vm}$ over the displaced copies, let $\mS_h$ be their overlap matrix, and let
\begin{equation*}
    \kappa_h = \frac{\sum_{\vj} |(\vc_h)_{\vj}|\, \|e_{\vm + h\vj}\|_{L^2}}{\| h^{-|\vn|}\Delta_h^{\vn} e_{\vm} \|_{L^2}}
\end{equation*}
be the relative condition number of the sum.
Then, as $h \to 0$,
\begin{align*}
    \|\vc_h\| &\geq h^{-|\vn|} , \\
    \lambda_{\min}(\mS_h) &= \mathcal{O}\big(h^{2|\vn|}\big) , \\
    \kappa_h &= \Theta\big(h^{-|\vn|}\big) .
\end{align*}
\end{proposition}

\begin{proof}
By \cref{eq:forward-diff}, $\vc_h = h^{-|\vn|}\vv$, where the entries of $\vv$ are signed binomial coefficients and the entry for $\vj = \vn$ is one, so $\|\vv\| \geq 1$.
By \cref{prop:gto-limit}(ii) in $L^2$,
\begin{align*}
    \vv^\top \mS_h \vv
    &= \|\Delta_h^{\vn} e_{\vm}\|_{L^2}^2 \\
    &= h^{2|\vn|}\, \|\partial_{\vm}^{\vn} e_{\vm}\|_{L^2}^2 + \mathcal{O}\big(h^{2|\vn| + 1}\big) .
\end{align*}
The Rayleigh quotient gives $\lambda_{\min}(\mS_h) \leq \vv^\top\mS_h\vv / \|\vv\|^2$, which is the second statement.
For the third, all copies have the same norm $\|e_{\vm}\|_{L^2}$, so the numerator of $\kappa_h$ is $h^{-|\vn|}\, \|\vv\|_1\, \|e_{\vm}\|_{L^2}$, while the denominator tends to $\|\partial^{\vn}_{\vm} e_{\vm}\|_{L^2} > 0$.
\end{proof}

\Cref{prop:gto-limit} licenses two statements.
First, the chart of \cref{eq:splat} needs no angular factor, since positive splats with signed coefficients reach every angular momentum.
Second, \cref{eq:gto-energy-bound} shows that a cloud is at least as expressive as the Gaussian basis it replaces.
The construction proves inclusion and is not the representation that the optimizer finds.
It uses $|T_1| = 4$ splats per $p$ primitive and $|T_2| = 10$ per $d$ primitive, more than the $2l + 1$ functions of a shell, while the clouds of \cref{sec:experiments} reach lower energies than the Gaussian basis at matched function count.
That gain comes from optimizing where the splats go.

\Cref{prop:gto-cost} shows that the limit is not attained by any finite cloud.
Reproducing an angular factor exactly requires coincident splats with coefficients of size $h^{-l}$, and their overlap matrix becomes singular as $h^{2l}$.
The occupied Gram matrix $\mG = \mC^\top\mS\mC$ of \cref{app:proofs-regeigh} can stay well conditioned in this limit, because the coefficients compensate for the near-singular overlap.
The compensation is a cancellation, however, and in floating-point arithmetic with unit roundoff $u$ it carries a relative error of order $u\,\kappa_h$.
A finite cloud therefore imitates angular momentum $l$ only down to displacements $h$ with $u\, h^{-l} \ll 1$, and this, not the floor of \cref{eq:regeigh}, is the limit.

\subsection{Floating-orbital models}\label{app:fsgo}

Optimizing the positions and shapes of Gaussian basis functions is an old idea, and the question is why it did not become standard practice.
This subsection shows that the earlier floating-orbital models are special cases of GS-DFT, and it states the four obstacles that stopped them as precise costs.
Each obstacle is removed by a specific component of GS-DFT.

Gaussians entered quantum chemistry because their integrals have closed forms.
\citet{Boys1950Gaussian} observed that the product of two Gaussians on different centers is a single Gaussian on a third, \cref{eq:gpt}, so every one- and two-electron integral reduces to elementary functions and the Boys function.
The price is the shape.
A Gaussian has no cusp at the nucleus and decays faster than an atomic orbital, so many Gaussians are needed per orbital, and standard basis sets pin fixed contractions of them to the atoms (\cref{def:gto-basis}).

\citet{Frost1967FSGO} took the opposite route with the floating spherical Gaussian orbital (FSGO) model.
It assigns one normalized spherical Gaussian to each doubly occupied orbital and minimizes the Hartree--Fock energy of the resulting determinant over the centers and widths, with no atom-centered basis.
\citet{Vescelius1974Ellipsoidal} and \citet{Cohen1976Ellipsoidal} made the Gaussians ellipsoidal.
\citet{Pederson1988FloatingGaussians} optimized floating Gaussians in DFT by simulated annealing and used the absence of Pulay terms to obtain Hellmann--Feynman forces without corrections, the result of \cref{prop:forces}.
Explicitly correlated Gaussians pushed nonlinear optimization of Gaussian parameters to spectroscopic accuracy, at a cost that confines them to a handful of particles \citep{Mitroy2013ECG}.
More recently, \citet{TamayoMendoza2018AutodiffHF} used automatic differentiation to optimize basis exponents and centers in Hartree--Fock on small molecules.

\begin{proposition}[FSGO is GS-DFT at $M = N_{\mathrm{occ}}$]\label{prop:fsgo}
Let $M = N_{\mathrm{occ}}$ and let $\mC$ be invertible.
Then the density matrix of GS-DFT does not depend on $\mC$,
\begin{equation*}
    \mP = 2\, \mS(\Theta)^{-1} ,
\end{equation*}
and the GS-DFT energy is a function of $\Theta$ alone.
With isotropic splats and $E_{\xc}$ set to exact exchange, it is the FSGO energy of \citet{Frost1967FSGO}.
With anisotropic splats, it is the ellipsoidal model of \citet{Vescelius1974Ellipsoidal}.
\end{proposition}

\begin{proof}
With $\mG = \mC^\top\mS\mC$ and the L\"owdin coefficients of \cref{eq:lowdin},
\begin{align*}
    \mP &= 2\, \bar{\mC}\bar{\mC}^\top \\
        &= 2\, \mC\, (\mC^\top \mS \mC)^{-1} \mC^\top \\
        &= 2\, \mC\, \mC^{-1} \mS^{-1} \mC^{-\top} \mC^\top \\
        &= 2\, \mS^{-1} .
\end{align*}
The density $\rho = \sum_{\mu\nu} \mP_{\mu\nu}\, g_\mu g_\nu$, and hence every term of the energy, then depends on $\Theta$ only.
It is the density of the determinant built from the $N_{\mathrm{occ}}$ Gaussians themselves, whose energy is invariant under any invertible recombination of its orbitals.
That is the determinant FSGO optimizes, with spherical Gaussians in the original model and ellipsoidal ones in its extension.
\end{proof}

\Cref{prop:fsgo} places the whole lineage in one corner of the GS-DFT family: the minimal cloud, $M = N_{\mathrm{occ}}$, with the coefficients reduced to a gauge.
GS-DFT keeps the model and moves along the other axis, $M > N_{\mathrm{occ}}$.
Four obstacles made that axis unreachable.
Write $D = M(9 + N_{\mathrm{occ}})$ for the number of free parameters and $T$ for the number of optimization steps.
\begin{enumerate}
\item[(i)] \emph{Integral cost.}
A fixed basis computes its electron-repulsion tensor once and reuses it at every iteration.
In a floating basis every entry depends on $\Theta$, so the two-electron term costs
\begin{equation*}
    \mathcal{O}(M^4) \text{ integrals per step, } \qquad \mathcal{O}(T M^4) \text{ in total.}
\end{equation*}
A fixed auxiliary basis $\{\varphi_a\}$ pinned to the atoms does not remove this cost without an accuracy penalty, because by \cref{eq:df-gap} its error is $\tfrac12 \|\rho - \Pi_\varphi \rho\|_{\mathrm{C}}^2$, and nothing controls it once the cloud, and with it $\rho$, moves away from $\operatorname{span}\{\varphi_a\}$.
\item[(ii)] \emph{Gradients.}
Optimizing $\Theta$ needs the gradient of the energy with respect to all $D$ parameters.
Finite differences cost $D + 1$ energy evaluations per gradient, and analytic gradients without automatic differentiation require a hand-derived derivative for every integral class.
\item[(iii)] \emph{Optimization at scale.}
The energy is nonconvex in $\Theta$.
Newton-type methods need the $D \times D$ Hessian,
\begin{equation*}
    \mathcal{O}(D^2) \text{ memory and } \mathcal{O}(D^3) \text{ time per step,}
\end{equation*}
which is infeasible for the clouds of \cref{sec:experiments}.
\item[(iv)] \emph{Linear dependence.}
Freely moving functions approach each other, $\lambda_{\min}(\mS) \to 0$ (\cref{prop:gto-cost}), and the eigendecomposition gradient of the orthonormalization contains factors $(\lambda_i - \lambda_j)^{-1}$ that diverge at degeneracies.
\end{enumerate}

\Cref{tab:fsgo-obstacles} pairs each obstacle with the component of GS-DFT that removes it.
The adaptive density fit of \cref{sec:method} rebuilds the auxiliary basis from the cloud, so its error is bounded by \cref{prop:pairproduct} rather than growing with the motion of the cloud, and its per-step cost is quadratic (\cref{alg:gsdft}).
Reverse-mode automatic differentiation evaluates the full gradient at a cost bounded by a constant multiple of one energy evaluation, independent of $D$ \citep{baydinAutomaticDifferentiationMachine2018}.
Adam \citep{kingmaAdamMethodStochastic2015} needs $\mathcal{O}(D)$ memory and time per step, the regime in which 3DGS optimizes millions of Gaussians \citep{Kerbl2023GaussianSplatting}.
The regularized eigensolve bounds the backward pass at degeneracies (\cref{prop:regeigh}).

\begin{table}[b]
\centering
\caption{Obstacles that confined floating-Gaussian models to $M = N_{\mathrm{occ}}$, and the component of GS-DFT that removes each.}
\label{tab:fsgo-obstacles}
\resizebox{\linewidth}{!}{
\begin{tabular}{llll}
\toprule
Obstacle & Earlier cost & Component of GS-DFT & Cost or bound \\
\midrule
(i) Integrals & $\mathcal{O}(M^4)$ per step & Adaptive density fitting & \cref{alg:gsdft}, \cref{prop:pairproduct} \\
(ii) Gradients & $D + 1$ evaluations & Reverse-mode differentiation & $\mathcal{O}(1)$ evaluations \\
(iii) Scale & $\mathcal{O}(D^3)$ per step & Adam & $\mathcal{O}(D)$ per step \\
(iv) Dependence & unbounded gradient & Regularized eigensolve & \cref{prop:regeigh} \\
\bottomrule
\end{tabular}}
\end{table}

None of the obstacles is a limit of the representation, which \cref{prop:gto-limit} shows is at least as expressive as a fixed basis.
They are limits of the numerics, and they are what the related-work section means by inaccurate and unscalable: the earlier models were inaccurate because they were minimal, and they stayed minimal because (i) to (iv) made $M > N_{\mathrm{occ}}$ too expensive.

\subsection{Correspondence with 3D Gaussian Splatting}\label{app:3dgs}

\Cref{sec:background} writes 3DGS and DFT as two instances of one optimization problem over Gaussian primitives, \cref{eq:3dgs,eq:ksdft}.
The correspondence is what lets GS-DFT inherit the 3DGS toolkit, and it is also easy to overstate.
This subsection proves what carries over exactly and states precisely where the two problems differ.

\paragraph{The parametrization carries over exactly.}
\citet{Kerbl2023GaussianSplatting} parameterize the covariance of each Gaussian by a unit quaternion $\vq$ and a log-scale vector $\bm{\sigma} \in \mathbbm{R}^3$,
\begin{equation*}
    \mSigma = \mU(\vq)\, \diag(e^{2\bm{\sigma}})\, \mU(\vq)^\top .
\end{equation*}
Since $\mA = \mSigma^{-1}$ and $\mU(\vq)$ is orthogonal,
\begin{align*}
    \mA &= \mU(\vq)\, \diag(e^{-2\bm{\sigma}})\, \mU(\vq)^\top , \\
    \vl &= -2\bm{\sigma} ,
\end{align*}
so the 3DGS parametrization and the parametrization of \cref{eq:splat} are the same map up to a fixed linear change of the log-scales.
Both are unconstrained, every point is a valid Gaussian, and both are equivariant (\cref{prop:equivariance}).
We parameterize the precision because it is what enters the closed-form integrals, through the product precision $\mA_\mu + \mA_\nu$ of \cref{eq:ovlp,eq:gpt}.
Both methods optimize all primitives jointly with Adam, although 3DGS assigns each parameter group its own learning rate while GS-DFT uses a single rate (\cref{app:recipe}).
\Cref{tab:3dgs} pairs the remaining components.

\begin{table}[t]
\centering
\small
\caption{Components of 3DGS and their counterparts in GS-DFT.}
\label{tab:3dgs}
\begin{tabular}{lll}
\toprule
 & 3DGS & GS-DFT \\
\midrule
Primitive & anisotropic Gaussian & anisotropic Gaussian \\
Shape chart & covariance, quaternion and log-scales & precision, quaternion and log-eigenvalues \\
Weights & opacity and color, nonnegative & orbital coefficients $\mC$, signed \\
Channels & color channels & occupied orbitals \\
Represented signal & radiance along rays & orbitals $\psi_k$ and density $\rho$ \\
Forward map & projection and alpha compositing & closed-form integrals and quadrature \\
Objective & photometric loss against views & variational energy, no data \\
Constraint & none & orthonormality, via \cref{eq:lowdin} \\
Initialization & structure-from-motion points & chemically informed, from nuclei \\
Cloud size & adaptive, by densification and pruning & fixed $M$ \\
Optimizer & Adam & Adam \\
\bottomrule
\end{tabular}
\end{table}

\paragraph{The model is linear where 3DGS is not.}
3DGS renders a pixel $\vx$ by sorting the Gaussians along its ray and alpha compositing their colors $\vc_i$,
\begin{align}
    \mathcal{R}(\vx) &= \sum_{i} \vc_i\, a_i(\vx) \prod_{j < i} \big(1 - a_j(\vx)\big) , \label{eq:alpha-blend} \\
    a_i(\vx) &= o_i\, G^{\mathrm{2D}}_i(\vx) , \nonumber
\end{align}
with opacity $o_i \in (0, 1)$ and projected footprint $G^{\mathrm{2D}}_i \leq 1$.
Expanding the product,
\begin{equation*}
    \mathcal{R}(\vx) = \sum_i \vc_i\, a_i(\vx) + \mathcal{O}\big(\max_j a_j(\vx)^2\big) ,
\end{equation*}
so the linear superposition of \cref{eq:3dgs} is the first-order, low-opacity limit of \cref{eq:alpha-blend}.
In GS-DFT the orbitals are exactly linear in $\mC$ at fixed $\Theta$ and the density is exactly quadratic.

\begin{proposition}[Signed coefficients are necessary]\label{prop:signed}
If $N_{\mathrm{occ}} \geq 2$, no two orthonormal orbitals can both be nonnegative combinations of splats.
\end{proposition}

\begin{proof}
Every splat is strictly positive everywhere, so a nonnegative, nonzero combination of splats is strictly positive everywhere.
If $\psi_1$ and $\psi_2$ were both of this form, then $\psi_1 \psi_2 > 0$ everywhere and
\begin{equation*}
    \int \dd[3]{\rr}\, \psi_1(\rr)\, \psi_2(\rr) > 0 ,
\end{equation*}
which contradicts orthogonality.
\end{proof}

At most one orbital can be sign-definite, and every other orbital needs coefficients of both signs.
This is the nodal structure that \cref{prop:gto-limit} builds from displaced splats of opposite sign.

\paragraph{The objective couples every pair of primitives.}
In 3DGS each pixel depends only on the Gaussians whose footprints cover it, and the loss is a sum over pixels.
In GS-DFT, the Hartree energy couples every pair of splat products through the Coulomb kernel,
\begin{equation*}
    E_H = \tfrac12 \sum_{\mu\nu\lambda\sigma} \mP_{\mu\nu} \mP_{\lambda\sigma}\, (g_\mu g_\nu \,|\, g_\lambda g_\sigma) ,
\end{equation*}
and this coupling does not decay to zero with distance.
By the Gaussian product theorem, $g_\mu g_\nu$ carries charge $\mS_{\mu\nu}$ at the center $\vm_{\mu\nu}$ of \cref{eq:gpt}, and as the two products separate,
\begin{equation*}
    (g_\mu g_\nu \,|\, g_\lambda g_\sigma) = \frac{\mS_{\mu\nu}\, \mS_{\lambda\sigma}}{|\vm_{\mu\nu} - \vm_{\lambda\sigma}|} \big(1 + o(1)\big) .
\end{equation*}
The interaction falls off only as the inverse distance, so it cannot be truncated by locality the way a rendering footprint can.
This is why GS-DFT needs the density fit of \cref{app:proofs-df}.
The orthonormality constraint couples every pair of orbitals in the same global way.

\paragraph{\rev{Quantum mechanics as the renderer.}}
Splatting in graphics is the projection of each Gaussian onto the image plane.
The elliptical weighted average filter of \citet{Zwicker2001SurfaceSplatting} replaces the camera map by its Jacobian $\mJ_{\mathrm{cam}}$ at the Gaussian's center, so that the projection stays Gaussian with covariance
\begin{equation*}
    \mSigma' = \mJ_{\mathrm{cam}}\, \mW\, \mSigma\, \mW^\top \mJ_{\mathrm{cam}}^\top ,
\end{equation*}
for a viewing transformation $\mW$.
The projection is therefore exact only to first order in the camera map.
\begin{revblock}
GS-DFT projects onto the $N$-dimensional occupied subspace and integrates in three dimensions.
All the integrals are exact closed forms, and the only approximation in the forward map is the quadrature of $E_{\xc}$.
The name thus refers to an analogy, and not to a one-to-one mapping of the procedure.
\end{revblock}

\paragraph{There is no data.}
3DGS solves an inverse problem.
It fits a set of observed views $\{I_v\}$, and its quality is measured on views held out from the fit, so the training loss can decrease while the held-out error increases.
GS-DFT solves a forward problem.
For exact integrals and orthonormal orbitals, the Rayleigh--Ritz principle bounds every parameter value by the complete-basis energy $E_{\mathrm{CBS}}$ of the same functional,
\begin{equation*}
    E(\Theta, \mC) \;\geq\; E_{\mathrm{CBS}} \qquad \text{for all } (\Theta, \mC) ,
\end{equation*}
so a lower energy is always a better solution, and there is no held-out set and no generalization gap.
With density fitting, the bound holds up to the gap of \cref{eq:df-gap}.
The price is that there is no data to initialize from, which is why GS-DFT places its initial cloud from the nuclei instead of from reconstructed points.

\section{Implementation}\label{appendix:implementation}
This appendix describes GS-DFT as it runs.
\Cref{app:software} gives the software stack and the numerical conventions, and \cref{app:integrals} derives every integral of the energy.
\Cref{app:complexity} gives the cost and memory of each operation, and \cref{app:sharding} explains how one system is spread over several devices.
\Cref{app:recipe} lists the training settings, \cref{app:evaluation} the evaluation protocol behind every reported number, and \cref{app:reproducibility} the hardware, software and code layout.

\subsection{Software and numerics}\label{app:software}

GS-DFT uses \texttt{dftax} \citep{GuzmanCordero2026dftax} as its base, because it implements Kohn--Sham DFT natively in JAX \citep{bradburyJAXComposableTransformations2018, kidgerEquinoxNeuralNetworks2021}.
The engine supplies the exchange--correlation functionals, the Becke quadrature grids, the Boys function, and the fixed-basis Gaussian reference calculations.
Everything specific to splats lives in the \texttt{gs-dft} package.\footnote{\ificlrfinal Repository\else Anonymized repository\fi: \coderepo.}
PySCF is used only as a reference for the calculations listed in \cref{app:evaluation}.

The parametrization is implemented as written in \cref{eq:splat}, so the precision and its derived quantities are closed forms in $(\vq, \vl)$,
\begin{align*}
    \mA &= \mU(\vq)\, \diag(e^{\vl})\, \mU(\vq)^\top , \\
    \det \mA &= e^{\sum_i \ell_i} , \\
    g_\theta(\rr) &= \det(\mA / \pi)^{1/4}\, \exp\!\big\{ -\tfrac12 (\rr - \vm)^\top \mA\, (\rr - \vm) \big\} ,
\end{align*}
and no matrix decomposition of a splat's own precision is ever needed.

\subsection{Integrals}\label{app:integrals}

Every term of \cref{eq:energy-terms} except the exchange--correlation (XC) energy reduces to a Gaussian integral over a product of two splats, followed, for the Coulomb terms, by a one-dimensional quadrature.
The XC energy is a quadrature on a grid.
The building block is the product of \cref{eq:gpt},
\begin{equation}
\label{eq:product-explicit}
    g_\mu(\rr)\, g_\nu(\rr) = N_\mu N_\nu\, \mK_{\mu\nu}\, \exp\!\big\{ -\tfrac12 (\rr - \vm_{\mu\nu})^\top \mA_{\mu\nu}\, (\rr - \vm_{\mu\nu}) \big\} ,
\end{equation}
with $N_\mu = \det(\mA_\mu / \pi)^{1/4}$, the prefactor $\mK_{\mu\nu}$ of \cref{eq:malahanobis}, and
\begin{align*}
    \mA_{\mu\nu} &= \mA_\mu + \mA_\nu , \\
    \vm_{\mu\nu} &= \mA_{\mu\nu}^{-1} (\mA_\mu \vm_\mu + \mA_\nu \vm_\nu) .
\end{align*}

\paragraph{Overlap and kinetic energy.}
Integrating \cref{eq:product-explicit} gives the overlap of \cref{eq:ovlp},
\begin{align*}
    \mS_{\mu\nu}
    &= N_\mu N_\nu\, \mK_{\mu\nu}\, (2\pi)^{3/2} \det(\mA_{\mu\nu})^{-1/2} \\
    &= 2^{3/2}\, \frac{(\det \mA_\mu \det \mA_\nu)^{1/4}}{\det(\mA_\mu + \mA_\nu)^{1/2}} \, \mK_{\mu\nu}
    = \sqrt \frac{\det \sqrt{\mA_\mu \mA_\nu}}{\det(\tfrac12 (\mA_\mu + \mA_\nu))} \, \mK_{\mu\nu} \;.
\end{align*}
For the kinetic energy, $\grad g_\mu(\rr) = -\mA_\mu (\rr - \vm_\mu)\, g_\mu(\rr)$.
Writing $\rr = \vm_{\mu\nu} + \vx$, with $\va = \vm_{\mu\nu} - \vm_\mu$ and $\vb = \vm_{\mu\nu} - \vm_\nu$, the integrand is a quadratic form in $\vx$ under a Gaussian of covariance $\mA_{\mu\nu}^{-1}$, and the cross terms vanish by symmetry.
Hence,
\begin{align*}
    \mT_{\mu\nu}
    &= \tfrac12 \int \dd[3]{\rr}\, \grad g_\mu \cdot \grad g_\nu \\
    &= \tfrac12\, \mS_{\mu\nu} \Big[ \operatorname{tr}\!\big( \mA_\mu \mA_\nu \mA_{\mu\nu}^{-1} \big) + \va^\top \mA_\mu \mA_\nu\, \vb \Big] .
\end{align*}

\paragraph{The Coulomb kernel.}
The Coulomb kernel has no Gaussian closed form, but it is a superposition of Gaussians,
\begin{equation}
\label{eq:laplace}
    \frac{1}{|\rr|} = \frac{2}{\sqrt{\pi}} \int_0^\infty e^{-t^2 |\rr|^2}\, \dd t .
\end{equation}
Let $p$ and $q$ be two Gaussian charge distributions with charges $Q_p$ and $Q_q$, centers $\vm_p$ and $\vm_q$, and precisions $\mA_p$ and $\mA_q$.
Substituting \cref{eq:laplace} into the Coulomb inner product of \cref{eq:coulomb-inner} gives
\begin{equation}
\label{eq:coulomb-quadrature}
    (p \,|\, q) = Q_p Q_q\, \frac{2}{\sqrt{\pi}} \int_0^\infty
    \det\!\big( \mathbbm{1} + 2t^2 \mSigma_{pq} \big)^{-1/2}
    \exp\!\Big\{ -t^2\, \bm{\Delta}^\top \big( \mathbbm{1} + 2t^2 \mSigma_{pq} \big)^{-1} \bm{\Delta} \Big\}\, \dd t ,
\end{equation}
with
\begin{align*}
    \mSigma_{pq} &= \mA_p^{-1} + \mA_q^{-1} , \\
    \bm{\Delta} &= \vm_p - \vm_q .
\end{align*}
The derivation uses the fact that if $\vx$ and $\vy$ are distributed as the normalized densities $p / Q_p$ and $q / Q_q$, then $\vx - \vy - \bm{\Delta}$ is a centered Gaussian with covariance $\mSigma_{pq}$.
The inner integral is then the expectation of $e^{-t^2|\vz + \bm{\Delta}|^2}$ over $\vz \sim \mathcal{N}(0, \mSigma_{pq})$, which is the integrand of \cref{eq:coulomb-quadrature}.
For a splat product, the charge is the overlap, $Q = \mS_{\mu\nu}$, which is the charge used in the far-field argument of \cref{app:3dgs}.

Every Coulomb quantity of the method is an instance of \cref{eq:coulomb-quadrature}. 
These are the auxiliary metric $\mQ$ and the fit vector $\vt$ of \cref{eq:adaptive-density-fitting}, whose auxiliary functions are splat products, and the exact Hartree energy used for evaluation.
In general $\mSigma_{pq}$ is a $3 \times 3$ matrix, and its determinant and inverse are adjugate closed forms.
Two limits follow.
If $\mSigma_{pq} = s\,\mathbbm{1}$ is isotropic, the integral is
\begin{equation*}
    (p \,|\, q) = Q_p Q_q\, \frac{\operatorname{erf}\!\big( |\bm{\Delta}| / \sqrt{2s} \big)}{|\bm{\Delta}|} ,
\end{equation*}
which is the Boys function of order zero.
A nucleus is the limit $\mA_q \to \infty$ of a Gaussian, with $\mSigma_{pq} = \mA_{\mu\nu}^{-1}$, $Q_q = -Z_a$ and $\vm_q = \RR_a$, which gives the nuclear attraction.
The half-line integral in \cref{eq:coulomb-quadrature} is mapped to the unit interval by $t = s/(1-s)$ and evaluated with a \rev{64-node} Gauss--Legendre rule,
\begin{align*}
    \int_0^\infty f(t)\, \dd t
    &= \int_0^1 f\!\Big( \frac{s}{1-s} \Big) \frac{\dd s}{(1-s)^2} \\
    &\approx \sum_{n=1}^{\rev{64}} \frac{w_n}{(1 - s_n)^2}\, f\!\Big( \frac{s_n}{1 - s_n} \Big) ,
\end{align*}
with nodes $s_n$ and weights $w_n$ of the rule on $[0, 1]$.
Range-separated hybrids such as $\omega$B97M-V need the long-range kernel as well, and it is the same superposition over a finite interval,
\begin{equation*}
    \frac{\operatorname{erf}(\omega |\rr|)}{|\rr|} = \frac{2}{\sqrt{\pi}} \int_0^\omega e^{-t^2 |\rr|^2}\, \dd t .
\end{equation*}
The integrand of \cref{eq:coulomb-quadrature} is unchanged, and only the node table differs, as we now use a 12-node Gauss--Legendre rule on $[0, \omega]$ instead of the half-line rule.

\paragraph{Exchange--correlation.}
The XC energy is the only term evaluated on a grid.
The density, its gradient and the kinetic-energy density are evaluated exactly at the points $\rr_g$ of a Becke grid with weights $w_g$,
\begin{align*}
    \rho(\rr_g) &= \sum_{\mu\nu} \mP_{\mu\nu}\, g_\mu(\rr_g)\, g_\nu(\rr_g) , \\
    E_\xc &\approx \sum_g w_g\, e_\xc\big( \rho(\rr_g),\, |\grad\rho(\rr_g)|^2,\, \tau(\rr_g) \big) ,
\end{align*}
where $e_\xc$ is the energy density of the functional supplied by \texttt{dftax}.
Generalized-gradient functionals use $\rho$ and $|\grad\rho|^2$, and meta-GGAs such as r$^2$SCAN also use the kinetic-energy density $\tau$.
The VV10 nonlocal correlation of $\omega$B97M-V \citep{Vydrov2010VV10} is a double sum over the grid,
\begin{align*}
    E_{\mathrm{c}}^{\mathrm{nl}} &\approx \sum_g w_g \rho_g \Big[ \beta + \tfrac12 \sum_{g'} w_{g'} \rho_{g'}\, \Phi_{gg'} \Big] , \\
    \Phi_{gg'} &= -\frac{3}{2\, u_{gg'}\, u_{g'g}\, (u_{gg'} + u_{g'g})} , \\
    u_{gg'} &= \omega_0(\rr_g)\, |\rr_g - \rr_{g'}|^2 + \kappa(\rr_g) ,
\end{align*}
with $\rho_g = \rho(\rr_g)$ and the local functions $\omega_0$ and $\kappa$ of VV10.

\subsection{Complexity and memory}\label{app:complexity}

\Cref{tab:complexity} lists the leading-order cost of every operation in one energy-and-gradient evaluation with 
$M$ denoting the number of splats,
$N_{\mathrm{occ}}$ the number of occupied orbitals,
$N_{\mathrm{at}}$ the number of nuclei,
$N_{\mathrm{aux}}$ the number of auxiliary functions,
$k$ the number of retained neighbors per splat,
$N_g$ the number of grid points,
$K$ the refresh period, and
$N_{g,\mathrm{c}}$ the number of grid points per chunk.
The number of retained pairs is $|\mathcal{P}| = \mathcal{O}(Mk)$.

\begin{table}[t]
\captionsetup{skip=4pt}
\centering
\footnotesize
\setlength{\tabcolsep}{3pt}
\caption{Leading-order cost of one energy-and-gradient evaluation. The factorization of the auxiliary metric is paid once per refresh and appears amortized over the $K$ steps it serves.}
\label{tab:complexity}
\resizebox{\textwidth}{!}{%
\begin{tabular}{llll}
\toprule
\textbf{Operation} & \textbf{Algorithm} & \textbf{Compute} & \textbf{Memory} \\
\midrule
Overlap and kinetic, $\mS\mC$ and $\mT\mC$ & Dense row blocks & $\mathcal{O}(M^2 N_{\mathrm{occ}})$ & $\mathcal{O}(M N_{\mathrm{occ}})$ \\
Occupied Gram and retraction, \cref{eq:regeigh} & Eigensolve of $\mG = \mC^\top\mS\mC$ & $\mathcal{O}(M N_{\mathrm{occ}}^2 + N_{\mathrm{occ}}^3)$ & $\mathcal{O}(N_{\mathrm{occ}}^2)$ \\
Nuclear attraction & Retained pairs & $\mathcal{O}(M k N_{\mathrm{at}})$ & $\mathcal{O}(M k)$ \\
Fitted Hartree, \cref{eq:adaptive-density-fitting} & Fit vector over retained pairs & $\mathcal{O}(M k N_{\mathrm{aux}} + N_{\mathrm{aux}}^2 + N_{\mathrm{aux}}^3 / K)$ & $\mathcal{O}(N_{\mathrm{aux}}^2)$ \\
Fitted exchange, hybrids only & One fit vector per orbital pair & $\mathcal{O}(M k N_{\mathrm{aux}} N_{\mathrm{occ}}^2 + N_{\mathrm{aux}}^2 N_{\mathrm{occ}}^2 + N_{\mathrm{aux}}^3)$ & $\mathcal{O}(N_{\mathrm{aux}}^2)$ \\
Exchange--correlation & Chunked grid & $\mathcal{O}(N_g M N_{\mathrm{occ}})$ & $\mathcal{O}(N_{g,\mathrm{c}} M)$ \\
VV10 nonlocal correlation, $\omega$B97M-V only & Chunked double grid sum & $\mathcal{O}(N_g^2)$ & $\mathcal{O}(N_{g,\mathrm{c}} N_g)$ \\
Exact Hartree, evaluation only & Streamed product pairs & $\mathcal{O}(M^4)$ & $\mathcal{O}(M^2)$ \\
Nuclear forces, once at the end & Reverse pass in $\RR$ & $\mathcal{O}(M k N_{\mathrm{at}} + N_{\mathrm{at}}^2)$ & $\mathcal{O}(M k)$ \\
\bottomrule
\end{tabular}}
\end{table}

The training step costs
\begin{equation}
\label{eq:step-cost}
    \mathcal{O}\Big( M^2 N_{\mathrm{occ}} + M k N_{\mathrm{aux}} + N_{\mathrm{aux}}^2 + N_g M N_{\mathrm{occ}} + \frac{N_{\mathrm{aux}}^3}{K} \Big)
\end{equation}
for semilocal functionals.
The last term is the one cubic operation of the method, the factorization of the auxiliary metric, amortized over $K$ steps.
Hybrid functionals add the exchange row, whose metric solve is repeated at every step, which is why they are run only on the smaller systems of \cref{sec:experiments}.

\subsection{Multi-device sharding}\label{app:sharding}

One system is spread over $n$ devices by sharding the work and replicating the parameters.
The parameters $(\Theta, \mC)$ occupy $\mathcal{O}(M N_{\mathrm{occ}})$ memory, \rev{and partitioning them would force communication inside every contraction, so they stay replicated; only their Adam moments are split, and each device updates its share of the parameters.}
Instead, each expensive term of \cref{tab:complexity} is written as a sum over an index set, and the index set is partitioned.

\paragraph{What is partitioned.}
\rev{Four} index sets are split into $n$ blocks, one per device $d$:
\begin{align*}
    \{1, \ldots, M\} &= R_1 \cup \cdots \cup R_n && \text{rows of } \mS\mC \text{ and } \mT\mC , \\
    \rev{\{1, \ldots, N_{\mathrm{aux}}\}} &\rev{{}= B_1 \cup \cdots \cup B_n} && \rev{\text{rows of } \mQ \text{ and its Cholesky factor} ,} \\
    \mathcal{P} &= \mathcal{P}_1 \cup \cdots \cup \mathcal{P}_n && \text{retained pairs} , \\
    \{\rr_g\} &= \mathcal{G}_1 \cup \cdots \cup \mathcal{G}_n && \text{grid points} .
\end{align*}
Each term is then a sum of per-device partials followed by one all-reduce.
For the occupied Gram matrix, the fit vector and the exchange--correlation energy,
\begin{align*}
    \mG &= \sum_{d=1}^n \mC_{R_d}^\top (\mS\mC)_{R_d} , &
    \vt &= \sum_{d=1}^n \vt^{(d)} , \\
    \vt^{(d)}_a &= (\rho_{\mathcal{P}_d} \,|\, h_a) , &
    E_\xc &\approx \sum_{d=1}^n \sum_{g \in \mathcal{G}_d} w_g\, e_\xc(\rr_g) ,
\end{align*}
and the kinetic and nuclear-attraction energies are split in the same way over rows and pairs.
The retraction of \cref{eq:regeigh} on the $N_{\mathrm{occ}} \times N_{\mathrm{occ}}$ matrix $\mG$ runs replicated on every device after the all-reduce.
\rev{At each refresh, every device builds its rows of $\mQ$ and the Cholesky factorization proceeds by blocks across the devices, so the metric solve of \cref{eq:adaptive-density-fitting} is a blocked substitution on the row-sharded factor.}


\subsection{Training parameters}\label{app:recipe}

\paragraph{Optimizer.}
We use Adam with global-norm gradient clipping, a linear warmup and a cosine decay,
\begin{align*}
    \eta_t &= \eta_{\max} \Big( 0.01 + 0.99\, \frac{t}{T_w} \Big) && t < T_w , \\
    \eta_t &= \eta_{\max} \Big( 0.01 + 0.99 \cdot \tfrac12 \Big[ 1 + \cos\!\Big( \pi\, \frac{t - T_w}{T - T_w} \Big) \Big] \Big) && t \geq T_w ,
\end{align*}
with peak rate $\eta_{\max} = 10^{-2}$, warmup $T_w = \min(\lfloor T/10 \rfloor, 200)$ for a budget of $T$ steps, and gradients clipped to global norm $1$.
\Cref{tab:recipe} lists the settings shared by all runs.

\begin{table}[t]
\captionsetup{skip=4pt}
\centering
\small
\setlength{\tabcolsep}{5pt}
\caption{Settings shared by every GS-DFT run. Settings that vary by experiment are in \cref{tab:recipe-exp}.}
\label{tab:recipe}
\begin{tabular}{lll}
\toprule
\textbf{Setting} & \textbf{Value} & \textbf{Reference} \\
\midrule
Optimizer & Adam, clip at global norm $1$, no weight decay & above \\
Learning rate & $10^{-2}$ peak, $10^{-4}$ floor, cosine & above \\
Screening & $\tau = 10^{-7}$ on the normalized overlap, see below & \cref{sec:method} \\
Pair-list padding & $2.0$ & \cref{app:sharding} \\
Fit regularizer & $\lambda = 10^{-8}$ & \cref{prop:dunlap} \\
Eigensolve & $\varepsilon = \delta = 10^{-4}$ & \cref{eq:regeigh} \\
Coulomb quadrature & 20 Gauss--Legendre nodes & \cref{app:integrals} \\
Precision & float64 & \cref{app:software} \\
\bottomrule
\end{tabular}
\end{table}

\paragraph{Grids.}
The exchange--correlation quadrature uses the Becke grids of \texttt{dftax}, indexed by a level $L$ that sets the number of radial shells and the Lebedev order per atom,
\begin{equation*}
    L = 2 : (50, 194) , \qquad L = 3 : (75, 302) , \qquad L = 5 : (110, 590) .
\end{equation*}
Training uses level 3, except on the protein-scale systems, which use level 2.
The observables of \cref{fig:observables} are evaluated on a common level-5 grid, for the Gaussian references and for GS-DFT alike.

\paragraph{Screening.}
We apply pair screening to every run of \cref{fig:observables}, and whenever $M \geq 300$.

\paragraph{Initialization.}
Two initializations are used, and \cref{tab:recipe-exp} states which run uses which:
\begin{enumerate}
\item[(i)] \emph{Atom-centered.}
The $M$ splats are divided evenly over the atoms, each atom's splats are isotropic, $\mA = 2\alpha\mathbbm{1}$, with Gaussian exponents $\alpha$ log-spaced between $e^{-1}$ and $e^{4.5}$, and the centers and log-eigenvalues receive Gaussian perturbations of scale $0.1$ and $0.02$.
The coefficients are drawn as $\mC_{\mu k} \sim \mathcal{N}(0, 0.1^2)$.
\item[(ii)] \emph{Chemically informed.}
A quarter of the splats are placed at bond midpoints, with bonds detected from covalent radii, and the rest on the atoms in proportion to the nuclear charge, with at least one per atom.
Atom-centered splats are isotropic, $\mA = 2\alpha\mathbbm{1}$, with Gaussian exponents $\alpha$ log-spaced over the range of the element's primitive exponents in def2-SVP.
Bond-centered splats have precision eigenvalue $10 / d^2$ across a bond of length $d$ and $4 / d^2$ along it, so they are elongated along the bond, and a bond that carries several splats spreads these values over a factor of $2$.
The same perturbations as in (i) break the remaining symmetry.
The coefficients are then either random, as in (i), or the minimal-basis guess: the tabulated atomic occupations of each element are projected onto the splats, and the $N_{\mathrm{occ}}$ natural orbitals of largest occupation are kept.
\end{enumerate}
The comparison between these choices is shown in \cref{appendix:ablations}.

\paragraph{Per-experiment settings.}
\Cref{tab:recipe-exp} gives the settings that differ between experiments.
The refresh period differs because the size ladder was tuned to $K = 10$ after the basis-accuracy ladders had been run at $K = 50$.
The protein-scale runs have individual budgets: $4{,}500$ steps for Trp-cage, $1{,}500$ for defensin, $1{,}000$ for insulin, $972$ for $\beta$-spectrin, which ran to its walltime limit, and $600$ for 3IFU.

\begin{table}[b]
\captionsetup{skip=4pt}
\centering
\small
\setlength{\tabcolsep}{4pt}
\caption{Settings that vary by experiment. $n_{\mathrm{ao}}$ is the function count of the named Gaussian basis.}
\label{tab:recipe-exp}
\resizebox{\textwidth}{!}{%
\begin{tabular}{llllllc}
\toprule
\textbf{Result} & \textbf{Cloud size} & \textbf{Init} & \textbf{Steps} & \textbf{$K$} & \textbf{Grid} & \textbf{Seeds} \\
\midrule
Observables, \cref{fig:observables} & ladder of $M$ & (i) & $12{,}000$ & $10$ & $3$ & $3$ \\
Accuracy per parameter, \cref{fig:isoparams} & ladder of $M$ & (i) & $12{,}000$ & $50$ & $3$ & $3$ \\
Functional coverage, \cref{tab:xc-coverage} & $M = n_{\mathrm{ao}}$(cc-pVTZ) & (i) & $48{,}000$ & $50$ & $3$ & $3$ \\
Ionization potentials, \cref{tab:gaps} & $M = n_{\mathrm{ao}}$(def2-QZVPP) & (i) & $12{,}000$ & $50$ & $3$ & $3$ \\
Anion convergence, \cref{fig:anions}(a) & ladder of $M$ & (i) & $12{,}000$ & $50$ & $3$ & $3$ \\
LiF dissociation, \cref{fig:anions}(b) & $M = n_{\mathrm{ao}}$(cc-pVTZ) & (i) & $4{,}000$ & $50$ & $3$ & $3$ \\
Protein-scale systems, \cref{tab:size-ladder} & $M = n_{\mathrm{ao}}$(cc-pVTZ) & (ii), minimal-basis $\mC$ & per system & $10$ & $2$ & $1$ \\
\bottomrule
\end{tabular}}
\end{table}

\subsection{Evaluation protocol}\label{app:evaluation}

\paragraph{Reported energies.}
Every reported GS-DFT energy on the molecules of \cref{fig:observables,fig:isoparams} and \cref{tab:xc-coverage,tab:gaps} is evaluated after training with the exact Hartree energy of case (iii) in \cref{app:integrals}, over the full pair list and without screening.
For these runs, the density-fitting error of \cref{app:proofs-df} affects the trained parameters but not the reported energy of those parameters, which is variational for semilocal functionals.
The protein-scale energies of \cref{tab:size-ladder} are the fitted training objective, because the $\mathcal{O}(M^4)$ exact evaluation is out of reach at $M \approx 2 \times 10^4$, and they carry the density-fitting error of \cref{app:proofs-df}.
For hybrid functionals and Hartree--Fock, the exchange term is evaluated with the fit on the converged auxiliary basis, because the exact four-index exchange tensor does not fit in memory at the cloud sizes where the comparison is made.

\paragraph{Ionization potentials.}
Direct minimization returns an occupied subspace, not canonical orbitals.
The orbital energies are recovered by one Fock build in the converged basis, followed by an $N_{\mathrm{occ}} \times N_{\mathrm{occ}}$ eigensolve in the occupied subspace, and the ionization potential is $-\varepsilon_{\mathrm{HOMO}}$ under Hartree--Fock \citep{koopmansUeberZuordnungWellenfunktionen1934}.
No further self-consistent iteration is run, since it would report a nearby stationary point instead of the minimized state.

\paragraph{Forces.}
Forces are the explicit partial derivative of \cref{prop:forces}, evaluated by one reverse pass in $\RR$.
By \cref{eq:force-residual}, their error is linear in the optimization residual, so checkpoints used for forces are first polished.
The polish continues training in three legs of $2000$ steps at learning rates $10^{-3}$, $3 \cdot 10^{-4}$ and $10^{-4}$, then solves for $\mC$ by a self-consistent-field iteration in the frozen cloud, accepted only if it lowers the energy, and finishes with $1000$ steps at $10^{-4}$.
The protein-scale forces of \cref{tab:size-ladder} are the explicit partial derivative at the training checkpoint, without this polish, reported as the root mean square over all Cartesian components,
\begin{equation*}
    F_{\mathrm{rms}} = \Big( \frac{1}{3 N_{\mathrm{at}}} \sum_{a,\, i} F_{a,i}^2 \Big)^{1/2} .
\end{equation*}

\paragraph{Limits.}
Two estimators of a basis-set limit are used.
The Gaussian limit of \cref{fig:observables} and \cref{tab:xc-coverage} is the three-point geometric extrapolation of the last three rungs $E_1, E_2, E_3$,
\begin{align*}
    r &= \frac{E_3 - E_2}{E_2 - E_1} , \\
    E_{\mathrm{CBS}} &= E_1 + \frac{E_2 - E_1}{1 - r} ,
\end{align*}
which is used when $0 < r < 1$; otherwise the largest rung is kept.
Applied pointwise to densities and forces, the same rule gives the field references of the second row of \cref{fig:observables}.
The limits of \cref{fig:isoparams} come from a power-law fit over all rungs of each ladder,
\begin{equation*}
    E_P - E_\infty \propto P^{-\alpha} ,
\end{equation*}
with $E_\infty$ constrained below every rung and chosen to maximize the coefficient of determination of the linear fit of $\log(E - E_\infty)$ against $\log P$.
For each $E_\infty$ the inner fit is linear and exact, so the estimator is a one-dimensional search.

\paragraph{Observables.}
For a density $\rho$ and a reference $\rho_{\mathrm{ref}}$ with $N$ electrons, both evaluated on the same grid, the density error is the total variation per electron,
\begin{equation*}
    \mathrm{TV} = \frac{1}{2N} \sum_g w_g\, \big| \rho(\rr_g) - \rho_{\mathrm{ref}}(\rr_g) \big| .
\end{equation*}
The force error is the largest absolute Cartesian component of the difference,
\begin{equation*}
    \Delta F = \max_{a,\, i}\, \big| F_{a,i} - F^{\mathrm{ref}}_{a,i} \big| .
\end{equation*}
Energy errors in \cref{fig:observables} are reported per atom, so that systems of different size share an axis.

\paragraph{References.}
Gaussian reference energies, densities and forces are computed with \texttt{dftax} on the same training grid as the GS-DFT run they are compared with, and the reference observables of \cref{fig:observables} are evaluated on the same level-5 grid as the GS-DFT ones.
Matched comparisons set $M = n_{\mathrm{ao}}$ of the named basis, so that both representations have the same number of functions and, by \cref{eq:param-overhead}, nearly the same number of parameters.

\subsection{Reproducibility}\label{app:reproducibility}
\rev{Runs used NVIDIA A100 and H100 GPUs with 80~GB of memory.}
Multi-GPU runs used four GPUs of one node.
Every molecular geometry used in \cref{sec:experiments} is provided in the repository, except the protein structures of \cref{tab:size-ladder}, which are available from the FMODB database \citep{takayaFMODBWorldFirst2021}.
\rev{All experiments can be replicated by the code provided in the associated repository.}

\section{Comparison with differentiable DFT frameworks}\label{appendix:frameworks}

We compare GS-DFT with the three open-source differentiable DFT codes: D4FT \citep{liD4FTDeepLearning2023} and MESS \citep{Helal2024MESS}, both written in JAX, and DQC \citep{KasimVinko2022DQC}, written in PyTorch.
PySCF \citep{Sun2020PySCF} is not differentiable and enters only as the reference energy every code is measured against.
All codes solve restricted Kohn--Sham DFT with the PBE functional on the same geometries.

We test each code in two steps.
The first is a correctness check on water in the cc-pVDZ basis, where every code should agree with PySCF.
The second is a size ladder in the cc-pVTZ basis: water, ethanol, alanine dipeptide, and alanine chains of $5$ and $15$ residues, followed by insulin in a $784$-atom structure, not the $947$-atom FMODB entry of \cref{tab:size-ladder}.
GS-DFT runs with as many splats as the cc-pVTZ basis has functions, $M = n_{\mathrm{ao}}$.
Every code uses the PBE functional, the level-3 PySCF grid or the closest grid the code exposes, and its own default optimizer and convergence criterion.
Where a code has no grid level of its own, its grid is built by PySCF at level 3, as MESS does by default.
Every code runs on one node with no time limit and climbs the ladder until its first out-of-memory error or crash.
\rev{GS-DFT uses one node with four H100 GPUs of $80$~GB each as in \cref{tab:size-ladder}, and D4FT and MESS use a single H100 GPU with $80$~GB since neither distributes over several devices.}
PySCF and DQC run on the CPU with all $64$ cores and $500$~GB of a CPU node.
We record the energy, the peak GPU memory, the number of iterations, and the total time.

\subsection{Capabilities}\label{app:frameworks-capabilities}

\Cref{tab:frameworks-capabilities} lists what each code supports.
A code is marked as supporting relevant bases when it reproduces the PySCF energy of water to within $0.1$~mHa in both cc-pVDZ and cc-pVTZ, the bases used in production chemistry (\cref{app:frameworks-water}).
MESS passes this test at STO-3G and 6-31G, and fails from cc-pVDZ upwards.
D4FT passes it in double precision, while its default single-precision mode returns NaN.
DQC runs on the CPU only, because its interface to the integral library converts GPU tensors to host arrays and fails when the molecule is placed on a GPU.
Among the differentiable codes, GS-DFT is the only one that distributes over several GPUs, and none of the others offers both density fitting and forces on a GPU.
The largest system that another differentiable code completes at cc-pVTZ is ethanol, with $9$ atoms, in DQC.

\begin{table}[ht]
\captionsetup{skip=4pt}
\centering
\small
\caption{\textbf{Capabilities of differentiable DFT codes.} PySCF as the non-differentiable reference. Relevant bases means agreement with PySCF to within $0.1$~mHa for water in cc-pVDZ and cc-pVTZ. The last column is the largest system, in atoms, completed at cc-pVTZ on one node before OOM.}
\label{tab:frameworks-capabilities}
\resizebox{\linewidth}{!}{
\begin{tabular}{lccccccr}
\toprule
\textbf{Code} & \textbf{GPU} & \textbf{Multi-GPU} & \textbf{Differentiable} & \textbf{DF} & \textbf{Forces} & \textbf{Relevant bases} & \textbf{Largest system} \\
\midrule
GS-DFT (ours)                         & \cmark & \cmark & \cmark & \cmark & \cmark & \cmark & $2{,}742$ \\
D4FT \citep{liD4FTDeepLearning2023}   & \cmark & \xmark & \cmark & \xmark & \xmark & \cmark & \rev{$3$} \\
MESS \citep{Helal2024MESS}            & \cmark & \xmark & \cmark & \xmark & \xmark & \xmark & -- \\
DQC \citep{KasimVinko2022DQC}         & \xmark\tnote{a} & \xmark & \cmark & \cmark & \cmark & \cmark & $9$ \\
\midrule
PySCF \citep{Sun2020PySCF}            & \xmark & \xmark & \xmark & \cmark & \cmark & \cmark & $153$\tnote{b} \\
\rev{GPU4PySCF \citep{Wu2025GPU4PySCF}} & \cmark & \cmark & \xmark & \cmark & \cmark & \cmark & \rev{$\geq 196$} \\
\bottomrule
\end{tabular}
}
\par\vspace{0.2em}
{\footnotesize $^a$The GPU path fails while building the molecule. $^b$Set by wall-clock time on the CPU, not by memory.}
\end{table}


\subsection{Accuracy and Scaling}\label{app:frameworks-water}

\Cref{tab:frameworks-water} reports the water check.
D4FT agrees with PySCF to $0.04$~mHa at cc-pVDZ and to $0.05$~mHa at cc-pVTZ.
MESS uses Cartesian basis functions, so we compare it against PySCF in the same Cartesian basis.
It agrees to $0.006$~mHa at STO-3G and at 6-31G, and then fails once the basis carries polarization functions.
The failure arrives with the $d$ shells of cc-pVDZ and not with the $s$ and $p$ shells of the smaller bases, which points at the integrals rather than at the optimizer.
GS-DFT at a matched count of $M = n_{\mathrm{ao}} = 24$ ends $66.6$~mHa above the cc-pVDZ reference.
This is the smallest rung in this work, and it lies below the crossing of the two scaling laws of \cref{fig:isoparams}. 

\begin{table}[ht]
\captionsetup{skip=4pt}
\centering
\small
\caption{Water, cc-pVDZ, PBE, level-3 grid. $\Delta$ is the energy minus the PySCF energy. Times are wall-clock seconds including compilation. Iterations are SCF cycles for PySCF and optimization steps for the others. $^a$Cartesian basis, compared against Cartesian PySCF ($-76.33470$~Ha).}
\label{tab:frameworks-water}
\begin{tabular}{lrrrr}
\toprule
\textbf{Code} & $E$ (Ha) & $\Delta$ (mHa) & \textbf{Iterations} & \textbf{Time (s)} \\
\midrule
PySCF                   & $-76.33344$ & --          & $7$         & $1.9$ \\
D4FT          & $-76.33340$ & $+0.038$    & $637$       & $128$ \\
MESS$^a$                & $-76.92339$ & $-588.7$    & $110$       & $9.7$ \\
DQC (CPU)               & $-76.33345$ & $-0.010$    & --          & $15.4$ \\
GS-DFT ($M = 24$)       & $-76.26686$ & $+66.6$     & $12{,}000$  & $154$ \\
\bottomrule
\end{tabular}
\end{table}


\Cref{tab:frameworks-scaling} reports the cc-pVTZ ladder.
MESS fails on water before its first iteration.
DQC reaches ethanol and then runs out of host memory on alanine dipeptide, because in its default mode it stores the four-index electron repulsion tensor, which for $n_{\mathrm{ao}} = 468$ is $384$~GB in double precision.
\rev{D4FT runs out of memory on ethanol.}
Every differentiable baseline therefore stops at ethanol or before, while GS-DFT runs the whole ladder.
GS-DFT climbs to insulin at $17{,}448$ functions in $18.4$~GB, and its peak memory grows by a factor of $102$ across the ladder while the basis grows by a factor of $300$.
However, GS-DFT is slower in wall-clock time than a converged SCF on these systems, because it optimizes the basis together with the orbitals where an SCF needs tens of iterations at a fixed basis.

\begin{table}[ht]
\captionsetup{skip=4pt}
\centering
\small
\caption{\textbf{The cc-pVTZ ladder on one node}. Memory is the peak GPU memory in GB. GS-DFT reports time per optimization step and PySCF time per SCF cycle, both in seconds; D4FT and DQC report the wall-clock time of the whole run, since neither exposes a per-iteration cost. GS-DFT uses $M = n_{\mathrm{ao}}$ splats. OOM marks an out-of-memory failure, and a dash marks a system not attempted after an earlier failure or, for PySCF, one that did not finish within a day of wall-clock time.}
\label{tab:frameworks-scaling}
\begin{tabular}{lrrccccc}
\toprule
& & & \multicolumn{2}{c}{\textbf{GS-DFT}} & \textbf{D4FT} & \textbf{DQC} & \textbf{PySCF} \\
\cmidrule(lr){4-5}
\textbf{System} & \textbf{Atoms} & $n_{\mathrm{ao}}$ & Mem & Time/step & Wall (s) & Wall (s) & Time/cycle \\
\midrule
Water                  & $3$   & $58$       & $0.18$ & $0.10$ & $474$   & $19.9$ & $0.7$ \\
Ethanol                & $9$   & $174$      & $1.37$ & $0.11$ & \rev{OOM} & $29.3$ & $1.2$ \\
$(\textrm{Ala})_2$     & $22$  & $468$      & $1.49$ & $0.15$ & --      & OOM    & $4.6$ \\
$(\textrm{Ala})_5$     & $53$  & $1{,}158$  & $1.76$ & $0.34$ & --      & --     & $22.9$ \\
$(\textrm{Ala})_{15}$  & $153$ & $3{,}358$  & $6.43$ & $2.06$ & --      & --     & $238$ \\
$\textrm{Insulin}$     & $784$ & $17{,}448$ & $18.43$ & $80.1$ & --      & --     & -- \\
\bottomrule
\end{tabular}
\end{table}

\section{Ablations}\label{appendix:ablations}

Each subsection varies one setting of \cref{app:recipe} and holds the rest at their production values.
Energies are totals in Hartree and differences are in mHa.

\vspace{-0.2\baselineskip}
\subsection{Regularized eigensolve}\label{app:abl-floor}
\vspace{-0.2\baselineskip}
The regularized eigensolve of \cref{prop:regeigh} keeps the eigendecomposition gradient finite when the occupied Gram matrix has degenerate eigenvalues, where the standard gradient divides by zero.
We train four symmetric molecules for $12{,}000$ steps with it and without it (\cref{tab:abl-floor}).
Without the regularization, methane and benzene stall more than $40$~mHa above the regularized result.
On \ce{CO2} and \ce{N2} the two agree to within $2$~mHa, as they do on a $45$-residue alanine chain at matched wall-clock time.
The regularization removes a failure on some symmetric molecules and costs nothing elsewhere.

\begin{table}[ht]
\captionsetup{skip=4pt}
\centering
\small
\caption{Regularized eigensolve against a plain eigendecomposition, at matched settings. $\Delta$ is the energy without the regularization minus the energy with it.}
\label{tab:abl-floor}
\begin{tabular}{lrrrr}
\toprule
\textbf{System} & $M$ & \textbf{Regularized} & \textbf{Plain} & $\Delta$ (mHa) \\
\midrule
\ce{CH4}  & $86$  & $-40.51841$  & $-40.47054$  & $+47.9$ \\
\ce{C6H6} & $264$ & $-232.09545$ & $-232.05306$ & $+42.4$ \\
\ce{CO2}  & $90$  & $-188.48252$ & $-188.48038$ & $+2.1$ \\
\ce{N2}   & $60$  & $-109.45544$ & $-109.45653$ & $-1.1$ \\
\bottomrule
\end{tabular}
\end{table}

\vspace{-0.2\baselineskip}
\subsection{Initialization}\label{app:abl-init}
\vspace{-0.2\baselineskip}
The chemically informed initialization of \cref{app:recipe} has four ingredients, which we add one at a time to a random atom-centered start: element exponent ranges from def2-SVP (S1), bond-centered splats (S2), bond-aligned anisotropy (S3), and splats allocated proportionally to nuclear charge (S4).
\Cref{tab:abl-init} reports water at $M = 48$ and ethanol at $M = 144$ after $4{,}000$ steps.
The full initialization ends $27$~mHa lower on water and $35$~mHa lower on ethanol than random placement, with S1 and S4 contributing most.
The seed spread falls from $11.7$ to $2.1$~mHa on water and from $4.6$ to $0.1$~mHa on ethanol, so one run is representative.

\begin{table}[ht]
\captionsetup{skip=4pt}
\centering
\small
\caption{Initialization ingredients, added cumulatively. Median energy over seeds after $4{,}000$ steps, with the seed spread in mHa.}
\label{tab:abl-init}
\begin{tabular}{llrrrr}
\toprule
 & & \multicolumn{2}{c}{\ce{H2O}, $M = 48$} & \multicolumn{2}{c}{\ce{C2H5OH}, $M = 144$} \\
\cmidrule(lr){3-4}\cmidrule(lr){5-6}
 & \textbf{Adds} & \textbf{Energy} & \textbf{Spread} & \textbf{Energy} & \textbf{Spread} \\
\midrule
S0 & random, atom-centered   & $-76.30164$ & $11.7$ & $-154.87106$ & $4.6$ \\
S1 & element exponent ranges & $-76.32034$ & $1.3$  & $-154.89223$ & $1.9$ \\
S2 & bond-centered splats    & $-76.32044$ & $1.3$  & $-154.89412$ & $6.0$ \\
S3 & bond-aligned anisotropy & $-76.32078$ & $2.3$  & $-154.89656$ & $1.7$ \\
S4 & charge-proportional allocation & $-76.32853$ & $2.1$ & $-154.90574$ & $0.1$ \\
\bottomrule
\end{tabular}
\end{table}

\vspace{-0.2\baselineskip}
\subsection{Refresh period}\label{app:abl-refresh}
\vspace{-0.2\baselineskip}
The pair list and the auxiliary basis are rebuilt from the current splats every $K$ steps and frozen in between (\cref{app:recipe}).
A small $K$ follows the moving cloud closely but pays for rebuilds often.
We vary $K$ alone on alanine dipeptide at $M = 468$ for $12{,}000$ steps at $\tau = 10^{-7}$ (\cref{tab:abl-refresh}).
Cost falls monotonically with $K$ and saturates early: rebuilding every step costs $76\%$ more than rebuilding every $200$ steps, and $K = 10$ is already within $8\%$ of the cheapest setting.
The energies span $86$~mHa without ordering in $K$, so at this budget $K$ sets cost and not accuracy.
We use $K = 10$.

\begin{table}[ht]
\captionsetup{skip=4pt}
\centering
\small
\caption{Refresh period $K$ on alanine dipeptide, $M = 468$, $12{,}000$ steps, $\tau = 10^{-7}$. Time is the total wall-clock time of the run on one A100.}
\label{tab:abl-refresh}
\begin{tabular}{rrrr}
\toprule
$K$ & \textbf{Energy} (Ha) & \textbf{Time} (s) & \textbf{ms / step} \\
\midrule
$1$   & $-495.60818$ & $9{,}629$ & $802$ \\
$5$   & $-495.59166$ & $6{,}427$ & $536$ \\
$10$  & $-495.63820$ & $5{,}878$ & $490$ \\
$50$  & $-495.57240$ & $5{,}580$ & $465$ \\
$200$ & $-495.55246$ & $5{,}462$ & $455$ \\
\bottomrule
\end{tabular}
\end{table}
\vspace{-0.3\baselineskip}

\vspace{-0.2\baselineskip}
\subsection{Screening threshold}\label{app:abl-screen}
\vspace{-0.2\baselineskip}
A density pair enters the Coulomb term only if its product has weight above $\tau$ (\cref{app:sharding}), which fixes both the cost of a step and the error in the Coulomb energy.
We repeat the run above at $K = 10$ and vary $\tau$, where $\tau = 0$ keeps every pair (\cref{tab:abl-screen}).
At $\tau = 10^{-8}$ screening already removes half of the pairs and $46\%$ of the wall-clock time, and four further decades remove only $12\%$ more.
The energies again span $63$~mHa without ordering in $\tau$, the same scale as the refresh sweep.
We use $\tau = 10^{-7}$, at the knee of the pair count.

\begin{table}[ht]
\captionsetup{skip=4pt}
\centering
\small
\caption{Screening threshold $\tau$ on alanine dipeptide, $M = 468$, $12{,}000$ steps, $K = 10$. Pairs is the number retained at the last refresh, out of $203{,}560$ unscreened.}
\label{tab:abl-screen}
\begin{tabular}{lrrrr}
\toprule
$\tau$ & \textbf{Energy} (Ha) & \textbf{Pairs} & \textbf{Retained} & \textbf{Time} (s) \\
\midrule
$0$         & $-495.61730$ & $203{,}560$ & $100\%$ & $11{,}025$ \\
$10^{-8}$   & $-495.59341$ & $102{,}450$ & $50\%$  & $6{,}310$ \\
$10^{-7}$   & $-495.63820$ & $94{,}658$  & $47\%$  & $5{,}878$ \\
$10^{-6}$   & $-495.65585$ & $86{,}488$  & $42\%$  & $5{,}503$ \\
$10^{-5}$   & $-495.61626$ & $76{,}734$  & $38\%$  & $5{,}019$ \\
\bottomrule
\end{tabular}
\end{table}

\vspace{-0.3\baselineskip}
\vspace{-0.5\baselineskip}
\section{Additional results}\label{appendix:additional}

\vspace{-0.2\baselineskip}
\subsection{Bare anions}\label{app:anions}
\vspace{-0.3\baselineskip}
\Cref{fig:anions}(a) shows the fluoride ion.
\Cref{tab:anions} extends it to the hydroxide ion, for each cardinal number $\zeta$, a cloud with the function count of the plain cc-pV$\zeta$Z basis, trained with the atom-centered initialization of \cref{app:recipe}, against the plain and the augmented Gaussian bases of the same cardinal number.
From TZ onwards, the cloud lies below the augmented basis of the same cardinal number for both anions, while using only the functions of the plain basis.
At QZ, it already lies below aug-cc-pV5Z, the largest augmented basis we compute.
The diffuse support that augmentation adds by hand is therefore found by the optimizer, which moves splats outward from the nuclei.

\begin{table}[t]
\captionsetup{skip=4pt}
\centering
\small
\setlength{\tabcolsep}{5pt}
\caption{Bare anions at matched function count. Energies in Hartree, restricted PBE. $M$ is the function count of the plain cc-pV$\zeta$Z basis. Bold marks a splat energy below the augmented basis of the same cardinal number.}
\label{tab:anions}
\begin{tabular}{llrrrrr}
\toprule
 & $\zeta$ & $M$ & cc-pV$\zeta$Z & aug-cc-pV$\zeta$Z & GS-DFT & Spread \\
\midrule
\ce{F-}  & D & $14$  & $-99.66562$ & $-99.77465$ & $-99.08264$ & $304.1$ \\
         & T & $30$  & $-99.74783$ & $-99.79967$ & $\mathbf{-99.80099}$ & $4.7$ \\
         & Q & $55$  & $-99.77711$ & $-99.80743$ & $\mathbf{-99.81283}$ & $0.3$ \\
         & 5 & $91$  & $-99.79765$ & $-99.81012$ & $\mathbf{-99.81330}$ & $0.1$ \\
\midrule
\ce{OH-} & D & $19$  & $-75.63289$ & $-75.73520$ & $-75.54983$ & $83.2$ \\
         & T & $44$  & $-75.70329$ & $-75.75426$ & $\mathbf{-75.75548}$ & $1.8$ \\
         & Q & $85$  & $-75.72893$ & $-75.76038$ & $\mathbf{-75.76674}$ & $4.4$ \\
         & 5 & $146$ & $-75.74814$ & $-75.76270$ & $\mathbf{-75.77348}$ & $0.2$ \\
\bottomrule
\end{tabular}
\vspace{-0.5\baselineskip}
\end{table}
\vspace{-0.5\baselineskip}

\subsection{Dissociation energy ladders}\label{app:dissociation}

\begin{figure}[b]
\captionsetup{skip=4pt}
\centering
\includegraphics[width=\linewidth]{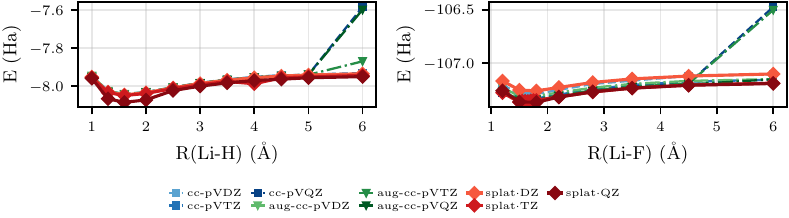}
\caption{Dissociation energy ladders $E(R)$ for LiH (left) and LiF (right), restricted PBE on the ionic \ce{Li+H-}/\ce{Li+F-} diabat. Plain cc-pVDZ/TZ/QZ (blue), augmented aug-cc-pVXZ (green), and splats sized to the plain cc-pVXZ counts (red).}
\label{fig:dissociation-both}
\end{figure}

\Cref{fig:dissociation-both} gives the total-energy dissociation curves $E(R)$ of LiH and LiF.
Each panel compares the plain cc-pV$\zeta$Z bases and their augmented aug-cc-pV$\zeta$Z counterparts, for $\zeta = \mathrm{D}, \mathrm{T}, \mathrm{Q}$, against splat clouds with the function counts of the plain bases.
Restricted PBE keeps both electrons of the bond paired, so as $R$ grows the closed-shell solution follows the ionic diabat \ce{Li+X-} rather than dissociating into neutral atoms.
This is the state every curve in \cref{fig:dissociation-both} is meant to describe, and at large $R$ it needs the diffuse density of an isolated anion, which is where the plain bases fail.
The Gaussian references are single-point restricted Kohn-Sham calculations from an atomic guess at each $R$.
At the most stretched geometries, the larger bases converge to a higher self-consistent solution instead of the ionic ground state.
These points are shown in \cref{fig:dissociation-both} but not used in any comparison.
On the other hand, the splat curves are computed as one continuation along $R$ rather than as independent calculations.
The first geometry starts from the atom-centered initialization, while each following geometry starts from the converged cloud of the previous one, splats closer to the displaced atom than to Li are translated with it, the coefficients are kept.
The cloud therefore follows the density from the covalent to the ionic regime instead of restarting at every geometry.
Each point is the lower of the continued and an independent calculation at that $R$.
The continued clouds track the ionic diabat below the entire ladder at matched function count.

\begin{revblock}
\subsection{Accuracy scaling across system sizes}\label{app:size-scaling}

\Cref{fig:size-scaling} overlays the accuracy ladders of \cref{fig:isoparams,fig:isoparams-intermediate} for systems from 7 to 196 atoms.
The exponent of each method stays nearly constant with system size, between $1.4$ and $1.7$ for the Gaussian ladders and between $3.9$ and $4.6$ for the splats.
\Cref{fig:scaling-summary} condenses these fits.
On the C/H/N/O systems, splats overtake the Gaussian ladder at $13.3 \pm 2.5$ basis functions per atom, reach $\Delta = 1$~mHa with $5$ to $11$ times fewer basis functions, and at the cc-pVTZ size their error per atom is $2$ to $6$ times smaller.
\ce{SF6} is the exception, with a later crossover and similar errors at the cc-pVTZ size.

\begin{figure}[t]
\captionsetup{skip=4pt}
\centering
\includegraphics[width=\linewidth]{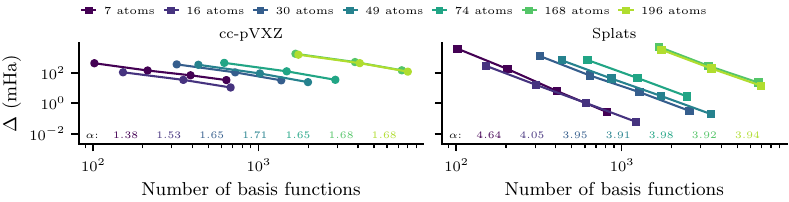}
\caption{\rev{\textbf{Accuracy scaling across system sizes.} Energy error $\Delta$ against the number of basis functions for the cc-pVXZ ladders (left) and the splat ladders (right), one line per system colored by its size, with the slope $\alpha$ of each line. Each ladder is measured from their respective limit.}}
\label{fig:size-scaling}
\end{figure}

\begin{figure}[b]
\captionsetup{skip=4pt}
\centering
\includegraphics[width=\linewidth]{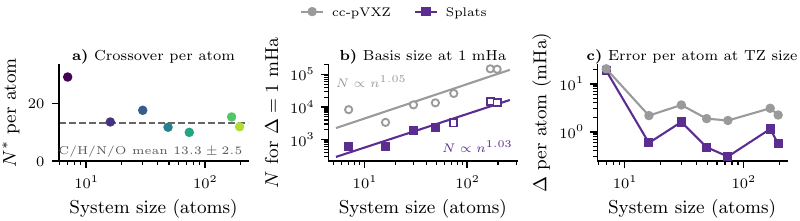}
\caption{\rev{\textbf{Summary of the accuracy scaling}, from the power law $\Delta = A N^{-\alpha}$ fitted to each ladder of \cref{fig:size-scaling}, against the system size $n$. \textbf{a)} Basis functions per atom at which the splat and Gaussian lines cross; the dashed line is the mean over the C/H/N/O systems. \textbf{b)} Basis functions needed for $\Delta = 1$~mHa, with a power law in $n$ per method; open markers are extrapolated beyond the measured ladder. \textbf{c)} Error per atom at each system's cc-pVTZ size.}}
\label{fig:scaling-summary}
\end{figure}

\subsection{Parallel efficiency}\label{app:parallel-efficiency}

\Cref{fig:parallel-efficiency} measures how the training step scales over the GPUs of one node.
The step time falls almost linearly with the number of GPUs: on four GPUs the speedup is $3.23$ for retinal and $3.90$ and $3.95$ for the two larger systems, whose computation outweighs the communication.
The peak memory of the busiest GPU falls by $1.9$ times on two GPUs and by $2.7$ to $2.9$ times on four.
The metric factor, the pair and cell lists, the integration grid and the optimizer state are split over the GPUs, while each GPU still holds the full coefficients and their gradient, which bounds the reduction.

\begin{figure}[t]
\captionsetup{skip=4pt}
\centering
\includegraphics[width=\linewidth]{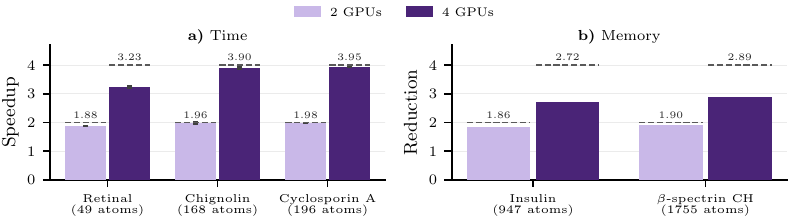}
\caption{\rev{\textbf{Parallel efficiency of the training step.} \textbf{a)} Speedup $t(1)/t(G)$ on $G$ H100 GPUs of one node, with error bars over repeated runs. \textbf{b)} Reduction of the peak memory of the busiest GPU, $m(1)/m(G)$, on $G$ A100 GPUs of one node. Dashed lines mark ideal scaling.}}
\label{fig:parallel-efficiency}
\end{figure}

\subsection{Observables across system sizes}\label{app:observables-scaling}

\Cref{fig:observables-sizes} repeats \cref{fig:observables} on six further systems with 7 to 74 atoms: \ce{SF6}, cubane and the adenine--thymine (A--T) base pair in its first two rows, and porphine, retinal and cholesterol in the last two.
As in \cref{fig:observables}, the observables converge as the clouds grow, decreasing monotonically on every system.
The values reached by the larger clouds, however, no longer translate into more accurate predictions, because they lie within the error of the PBE functional itself.
Beyond this point, a smaller representation error brings the result closer to the exact PBE solution, not to the true electronic structure.

\begin{figure}[b]
\captionsetup{skip=4pt}
\centering
\includegraphics[width=\linewidth]{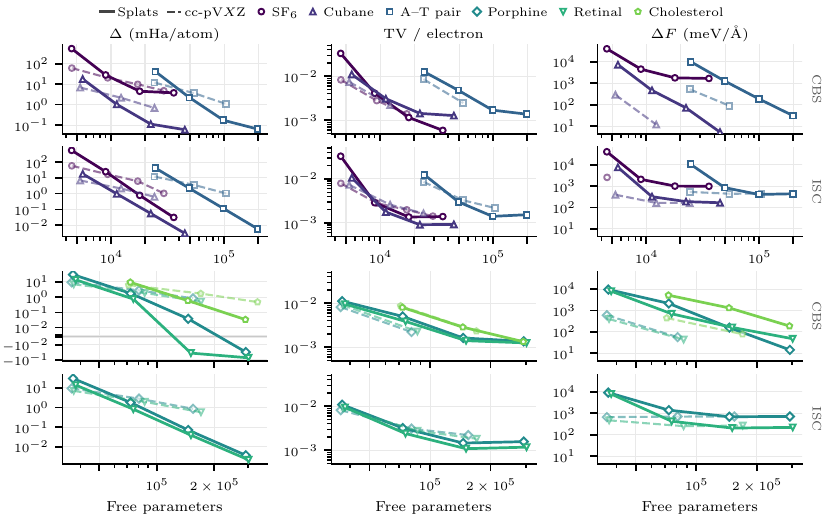}
\caption{\rev{\textbf{Observables across system sizes.} As in \cref{fig:observables}, one color per system, with the splats solid and the cc-pVXZ ladder dashed. \textbf{Rows 1 and 2:} \ce{SF6}, cubane and the A--T pair. \textbf{Rows 3 and 4:} porphine, retinal and cholesterol. The first row of each pair (CBS) measures the energy against the CBS limit and the density and forces against cc-pVQZ; the second (ISC) measures both ladders against the ISC limit, a power law over the four splat clouds.}}
\label{fig:observables-sizes}
\end{figure}
\end{revblock}

\end{document}